\documentclass[journal]{IEEEtran}

\usepackage{dsfont}
\usepackage{epsfig,latexsym,amsmath,amssymb,cite,color,graphicx,algorithm,algorithmic,cases}
\usepackage{amsmath,amssymb,amsfonts,hyperref,cite}
\usepackage{amsthm}
\usepackage{algorithmic}
\usepackage{graphicx}
\usepackage{setspace}
\usepackage[percent]{overpic}
\usepackage{stmaryrd}
\usepackage{bbm}
\usepackage{bm}
\usepackage{amsmath}
\usepackage{graphicx}
\usepackage[utf8]{inputenc}
\usepackage[english]{babel}
\usepackage{todonotes}

\newcommand{\mathdash}{\relbar\mkern-12mu\relbar}
\newtheorem{theorem}{Theorem}
\newtheorem{corollary}{Corollary}
\newtheorem{lemma}{Lemma}{}
\newtheorem{definition}{Definition}
\newtheorem{proposition}{Proposition}
\newtheorem{example}{Example}
\newtheorem{remark}{Remark}
\usepackage{soul}

\DeclareMathOperator*{\argmax}{arg\,max}
\DeclareMathOperator*{\argmin}{arg\,min}
\allowdisplaybreaks
\hypersetup{
    colorlinks=true, 
    linkcolor= blue,
    citecolor=blue 
}

\title{Secret Sharing from Correlated Gaussian Random Variables and Public Communication}

\author{Vidhi Rana, R\'{e}mi A. Chou, Hyuck Kwon \thanks{The authors are with the Department of Electrical
Engineering and Computer Science, Wichita State University, Wichita, KS. R\'{e}mi A. Chou and Vidhi Rana were supported in part by NSF grant CCF-1850227 and CCF-2047913. Part of results has been presented at the 2020 IEEE Information Theory Workshop \cite{rana2021secret}. E-mails: vxrana@shockers.wichita.edu, remi.chou@wichita.edu, hyuck.kwon@wichita.edu. }
}
\begin{document}

\maketitle
\begin{abstract}
    In this paper, we study an {\color{black} information-theoretic secret sharing problem}, where a dealer distributes shares of a secret among a set of participants under the following constraints: (i) authorized sets of users can recover the secret by pooling  their shares, and (ii) non-authorized sets of colluding users cannot learn any information about the secret. We assume that the dealer and participants observe the realizations of correlated Gaussian random variables and that the dealer can communicate with participants through a one-way, authenticated, rate-limited, and public channel. Unlike traditional secret
sharing protocols, in our setting, no perfectly secure channel is needed between the dealer and the~participants. Our main result is a closed-form characterization of the fundamental trade-off between secret rate and public communication rate.  
\end{abstract}

\section{Introduction}
Secret sharing has been introduced in \cite{shamir}, \cite{blakley}. In basic secret-sharing models, a dealer distributes a secret among a set of participants, with the constraint that only pre-defined sets of participants can recover this secret by pooling  their shares, while any other set of colluding participants cannot learn any information about the secret.

{\color{black} In most secret-sharing models, including Shamir's scheme~\cite{shamir}, it is assumed that the dealer and each participant can communicate over an information-theoretically secure channel at no cost.  
 While complexity-based cryptography techniques, e.g., \cite{Diffie1976}, could be used to implement secure channels without any other resources than a public channel, it would not provide information-theoretically secure channels. In this paper, we are interested in  
\emph{another approach that aims at providing a full information-theoretic solution that would not rely on complexity-based cryptography}.} In other words, we want to avoid the assumption that  information-theoretically secure communication channels are available at no cost. An information-theoretic approach to secret sharing over wireless channels  has been introduced in \cite{zou} for this purpose. The main idea is to leverage channel noise by remarking that {\color{black} information-theoretic secret sharing} over wireless channels is similar to compound wiretap channel models \cite{liang09}. This information-theoretic approach has also been formulated for source models in~\cite{csiszar2,profchou,chou2021}, where participants and dealers share correlated random variables. These models are related to compound secret-key generation, e.g., \cite{Tavangaran,bloch10}, and biometric systems with a multiuser access structure \cite{chou2019biometric}, in that multiple reliability and security constraints need to be satisfied simultaneously.

In this paper, we consider the {\color{black}information-theoretic} secret sharing model in \cite{profchou} with Gaussian sources. Specifically, the dealer and the participants observe realizations of correlated Gaussian random variables, and the dealer can communicate with the participants over an authenticated, one-way, rate-limited, and public communication channel.  
In wireless networks, independently and identically distributed realizations of correlated random variables can, for instance, be obtained from channel gain measurements after appropriate manipulations~\cite{ye,pierrot}.
Our approach for the achievability part consists in handling the reliability and security requirements separately. Specifically,  reliability is obtained via a coding scheme akin to a compound version of Wyner-Ziv coding \cite{wyner}, and security relies on universal hashing via extractors~\cite{vadhan}. Interestingly, the converse shows that there is no loss of optimality in decoupling the reliability and security requirements. The achievability is first obtained for discrete random variables and then extended to continuous random variables via fine quantization. In principle, one cannot assume a specific quantization strategy to ensure the security requirement in an information-theoretic manner; hence, the key step in this extension is to show that information-theoretic security holds, provided that the quantization is sufficiently fine. For the converse part, we can partly rely on techniques developed in \cite{r8}, \cite{r3}. However, unlike in \cite{r8}, \cite{r3}, our setting involves multiple security constraints that need to be satisfied simultaneously; hence, the main task in the converse is to prove a saddle point property without any degradation assumption on the source model.

The main differences between our work and \cite{Tavangaran,bloch10,chou2019biometric,profchou} are that \cite{Tavangaran,bloch10,chou2019biometric,profchou} consider discrete memoryless sources, whereas we consider Gaussian sources. As described above, handling Gaussian random variables calls for different proof techniques and considerations. Additionally, unlike \cite{Tavangaran,bloch10,chou2019biometric,profchou}, it also allows us to derive capacity results without assuming any source degradation properties. We also highlight that unlike~\cite{bloch10,profchou}, we consider rate-limited  public communication, and unlike~\cite{chou2019biometric,profchou}, we handle arbitrary access structures.

The main features of our work can be summarized as follows: 
(i) Our model relies on correlated Gaussian random variables and, similar to \cite{profchou} but unlike traditional secret-sharing schemes \cite{shamir}, does not rely on the assumption that information-theoretically secure channels between the dealer and the participants are available. (ii) Similar to the model in~\cite{profchou} but unlike traditional secret-sharing models, we consider a model that requires information-theoretic security for the secret with respect to unauthorized sets of participants during the distribution phase, i.e., when the dealer distributes shares of the secret to participants. 
(iii)~We establish a closed-form expression that characterizes the optimal trade-off between secret rate and public communication rate.
(iv)~The size of the shares in our coding scheme scales linearly with the size of the secret for any access structure similar to the model in~\cite{profchou}. Indeed, a share comprises the public communication from the dealer and $n$ quantized realizations of a Gaussian random variable, which can be shown to both linearly scale with $n$. The size of the shares does depend on the specific access structure considered but not on the number of participants. Specifically, the public communication must ensure that the set of authorized users with the least amount of information about the secret is able to reconstruct the secret.  By contrast, the best-known traditional secret-sharing schemes may require a share size that grows exponentially with the number of the participants for some access structures \cite{beimel} --  note, however, that it is unknown whether or not there exist traditional secret-sharing schemes that require a smaller share size. (v)~For threshold access structures, i.e., when a fixed number of participants $t$ is needed to reconstruct the secret (independently from the specific identities of those participants), we establish that the size of the secret that can be exchanged is, in general, \emph{not} a monotonic function of the threshold $t$.

The remainder of the paper is organized as follows. We set the notation in Section \ref{sec:notation} and  formally introduce the problem statement in Section \ref{secps}.  We present our main results in Section~\ref{secmr}, and  proofs in Sections~\ref{sec:conv} and \ref{sec:ach}. Finally, we provide concluding remarks in Section~\ref{sec:cm}. 

\section{Notation} \label{sec:notation}
For any $a,b\in \mathbb{R}$, define $\llbracket a,b\rrbracket \triangleq [ \lfloor a \rfloor, \lceil b \rceil]\cap \mathbb{N}$. For $x \in \mathbb{R}$, define $[x]^+\triangleq\max(0,x)$. For a set $\mathcal{S}$, let $2^{\mathcal{S}}$ denote the power set of $\mathcal{S}$. All logarithms are taken in base 2 throughout the paper. Let $I_m$ denote the identity matrix of  dimension $m \in \mathbb{N}$. Let $\det(W)$ denote the determinant of a matrix $W$ and $\vert \mathcal{S} \vert$ denote the cardinality of a set $\mathcal{S}$. For two random variables $X$ and $V$, $\sigma^2_X$ and $\sigma_{XV}$ denote $\mathbb{E}[(X-\mathbb{E}[X])^2]$ and $\mathbb{E}[(X-\mathbb{E}[X])(V-\mathbb{E}[V])]$, respectively. $N \sim \mathcal{N}(0,\Sigma)$ indicates that $N$ is a zero-mean Gaussian random vector with covariance matrix $\Sigma$. The indicator function is denoted by $\mathds{1}\{ \omega \}$, which is equal to $1$ if the predicate $\omega$ is true and $0$ otherwise. Let $H(X)$ (respectively, $h(X)$) denote the Shannon entropy (respectively, the differential entropy) of a discrete (respectively  continuous), random variable $X$. Also, let $I(X;Y)$ denote the mutual information between $X$ and $Y$, which are either continuous or discrete random variables. 
 
\section{Problem Statement} \label{secps}
Consider a dealer and $L$ participants. Define $\mathcal{L}\triangleq \llbracket 1, L \rrbracket$,   $\mathcal{X}\triangleq \mathbb{R}$, and $\mathcal{Y}\triangleq \mathbb{R}$. Consider a Gaussian memoryless source model $(\mathcal{X}\times \mathcal{Y}_{\mathcal{L}}, p_{XY_{\mathcal{L}}})$, where $Y_{\mathcal{L}}\triangleq (Y_l)_{l \in {\mathcal{L}}}$, and $(X, Y_\mathcal{L})$ are jointly Gaussian random variables with a non-singular covariance matrix. Let $\mathbb{A}$ be a set of subsets of $\mathcal{L}$ such that for any $\mathcal{T}\subseteq{\mathcal{L}}$, if $\mathcal{T}$ contains a set that belongs to $\mathbb{A}$, then $\mathcal{T}$ also belongs to $\mathbb{A}$, i.e., $\mathbb{A}$ is a monotone access structure~\cite{benaloh1988generalized}. We also define $\mathbb{U}\triangleq2^{\mathcal{L}}\backslash\mathbb{A}$ as the set of all colluding subsets of users who must not learn any information about the secret. In the following, for any $\mathcal{A}\in \mathbb{A}$ and for any $\mathcal{U}\in \mathbb{U}$, we use the notation  $Y_{\mathcal{A}}^n\triangleq (Y^n_l)_{l\in \mathcal{A}}$ and $Y_{\mathcal{U}}^n\triangleq (Y^n_l)_{l\in \mathcal{U}}$. Moreover, we assume that the dealer can communicate with the participants over an authenticated, one-way, rate-limited, noiseless, and public communication channel.  
\begin{definition}
A $(2^{nR_s}, R_p, \mathbb{A}, n)$ secret-sharing strategy is defined as follows:
\begin{itemize}
    \item The dealer observes $X^n$ and Participant $l \in \mathcal{L}$ observes~$Y^n_l$.
    \item The dealer sends over the public channel the message $M$ to the participants with the bandwidth constraint  $H(M)\leq n R_p$.
    \item The dealer computes a secret $S \in \mathcal{S} \triangleq \llbracket 1 , 2^{nR_s} \rrbracket$ from~$X^n$.
    \item Any subset of participants $\mathcal{A} \in \mathbb{A}$ can compute an estimate $\widehat{S}(\mathcal{A})$ of $S$ from their observations $(Y^n_l)_{l \in \mathcal{A}}$ and~$M$.
\end{itemize}
\end{definition}
\begin{definition}
A rate pair $(R_p,R_s)$ is achievable if there exists a sequence of $(2^{nR_s},R_p,\mathbb{A},n)$ secret-sharing strategies such~that
\begin{align}
\lim_{n \rightarrow{\infty}}\max_{\mathcal{A} \in \mathbb{A}}  \mathbb{P}[\widehat{S}(\mathcal{A})\neq S]&=0,\label{oct19_29_1}\\
\lim_{n \rightarrow{\infty}}\max_{\mathcal{U} \in \mathbb{U}} I(S;M,Y^n_\mathcal{U})&=0,\label{oct19_29_2}\\
\lim_{n \rightarrow{\infty}}\log\vert\mathcal{S}\vert-H(S)&=0\label{oct19_29_3}. 
\end{align}
\end{definition}
(\ref{oct19_29_1}) means that any subset of participants in
$\mathbb{A}$ is able to recover the secret,  (\ref{oct19_29_2}) means that any subset of participants in $\mathbb{U}$ cannot obtain information about the secret, while  (\ref{oct19_29_3}) means that the secret is nearly uniform and that its entropy is nearly equal to its length.
\begin{remark}
The uniformity condition (\ref{oct19_29_3}) ensures that a secret-sharing strategy that maximizes the length of the secret, will also maximize the entropy of the secret. Without this condition, maximizing the length of the secret would not be meaningful as one could always increase the length of the secret by adding redundancy to it. This is the same reason why in secret-key generation, one requires uniformity of the secret key \cite{maurer1993,ahls1993}.
\end{remark}
The secret capacity region is defined as $$ \mathcal{R}(p_{XY_{\mathcal{L}}}, \mathbb{A})\triangleq\{(R_p,R_s): (R_p, R_s)~\text{is achievable}\}.$$ Moreover, for a fixed $R_p$, the supremum of secret rates $R_s$ such that $(R_p,R_s) \in \mathcal{R}(p_{XY_{\mathcal{L}}}, \mathbb{A})$ is called the secret capacity and is denoted by $C_s(\mathbb{A}, R_p)$.

Additionally, one can write for any $\mathcal{A}\in \mathbb{A}$ and for any $\mathcal{U}\in \mathbb{U}$ (see Appendix \ref{App_th1} for the derivation)
\begin{align}
Y_{\mathcal{A}}&=H_{\mathcal{A}}X+W_{Y_{\mathcal{A}}},\label{aug_291}\\
Y_{\mathcal{U}}&=H_{\mathcal{U}}X+W_{Y_{\mathcal{U}}},\label{aug_292}
\end{align}
where $H_{\mathcal{A}} \in \mathbb{R}^{{\vert \mathcal{A}\vert}\times 1}$, $H_{\mathcal{U}}\in \mathbb{R}^{\vert{\mathcal{U}}\vert\times 1}$, $W_{Y_{\mathcal{A}}}\sim\mathcal{N}(0,I_{\vert\mathcal{A} \vert})$, and  $W_{Y_{\mathcal{U}}}\sim\mathcal{N}(0,I_{\vert \mathcal{U}\vert})$.
\section{Main Results}\label{secmr}
\subsection{Results for general access structures}
For a given access structure $\mathbb{A}$, define  $$\mathcal{A}^\star\in\arg\min_{\mathcal{A} \in \mathbb{A}}  H_\mathcal{A}^TH_\mathcal{A} , \quad \mathcal{U}^\star\in\arg\max_{\mathcal{U} \in \mathbb{U}}  H_\mathcal{U}^TH_\mathcal{U} .$$
\begin{theorem} \label{theo1}
For any access structure $\mathbb{A}$ and public communication rate $R_p\geq 0$, the secret capacity $C_s(\mathbb{A}, R_p)$ is
\begin{align*}
&C_s(\mathbb{A}, R_p)=\\ &\Bigg[\frac{1}{2}\log\frac{\sigma^2_XH^T_{\mathcal{U}^\star} H_{\mathcal{U}^\star} 2^{-2R_p}+\sigma^2_X H_{\mathcal{A}^\star}^T H_{\mathcal{A}^\star}  (1-2^{-2R_p})+1}{\sigma^2_XH^T_{\mathcal{U}^\star} H_{\mathcal{U}^\star}\!+1}\Bigg]^+\!\!\!.
\end{align*}
\end{theorem}
\begin{proof}
The converse and achievability are proved in Sections~\ref{sec:conv} and \ref{sec:ach}, respectively. \end{proof}
From Theorem \ref{theo1}, we obtain the following corollary when the public communication is rate-unlimited.
\begin{corollary} \label{theo2}
For any access structure $\mathbb{A}$, and an unlimited public communication rate, the secret capacity is given by
\begin{align*}
C_s(\mathbb{A}, R_p = +\infty ) &\triangleq \lim_{ R_p \to +\infty}C_s(\mathbb{A}, R_p )\\&=\Bigg[\frac{1}{2}\log\frac{\sigma^2_{X}H_{\mathcal{A}^\star}^TH_{\mathcal{A}^\star}+1}{\sigma^2_{X}H^T_{\mathcal{U}^\star}H_{\mathcal{U}^\star}+1}\Bigg]^+.
\end{align*}
\end{corollary}
{\color{black}

Note that in Theorem \ref{theo1} and Corollary \ref{theo2}, the length of the public communication scales linearly with the length of the secret by construction and corresponds to a compressed version of the $n$ source observations of the dealer via a compound version of Wyner-Ziv coding. Hence, the size of the share of each participant, which
comprises the public communication and $n$ quantized observations of a Gaussian random
variable, scales linearly with the length of the secret -- as explained in the proof of Theorem \ref{theo1}, the number of bits needed to store quantized realizations of Gaussian random variables is negligible compared to the number of source observations $n$ in our achievability scheme.  
Note that, unlike traditional secret-sharing models, which separately consider the share-creation phase and the share-distribution phase, we allow  a joint design of these two phases in our setting. This is made possible by considering   correlated random variables (at the participants and the dealer) and public communication  instead of information-theoretically secure channels as in traditional secret-sharing models.} 
The following example illustrates Theorem \ref{theo1} and Corollary~\ref{theo2}.
\begin{example} \label{ex1}
Consider a dealer and three participants who observe independently and identically distributed realizations of correlated Gaussian random variables as depicted in Figure~\ref{fig1}.  Define the access structure $\mathbb{A}\triangleq\{\{1,2\}, \{2,3\}, \{1,2,3\}\}$ and define  $\mathbb{U}\triangleq\{ \{1,3\}, \{1\}, \{2\}, \{3\}\}$ such that (i) the sets of participants in $\mathbb{A}$ can recover the secret using their observations and the public message $M$, and (ii) the sets of participants in $\mathbb{U}$ cannot learn information about the secret.
For $s \in \llbracket 1 ,L \rrbracket$, let $H_{\mathcal{L}}(s)$ denote the $s$-th component of $H_{\mathcal{L}}$, and assume that $\sigma^2_X\triangleq2$, $H_\mathcal{L}\triangleq[0.5, 1, 0.8]^T$, and for any $\mathcal{S} \subseteq \mathcal{L}$, $H_\mathcal{S} = (H_\mathcal{L}(s))_{s \in \mathcal{S}}$. Then, one can compute the secret capacity using Theorem~\ref{theo1} and Corollary \ref{theo2}, as shown in Figure~\ref{fig2}.
\end{example}
\begin{figure}[ht]
\centering
\includegraphics[width=0.49\textwidth]{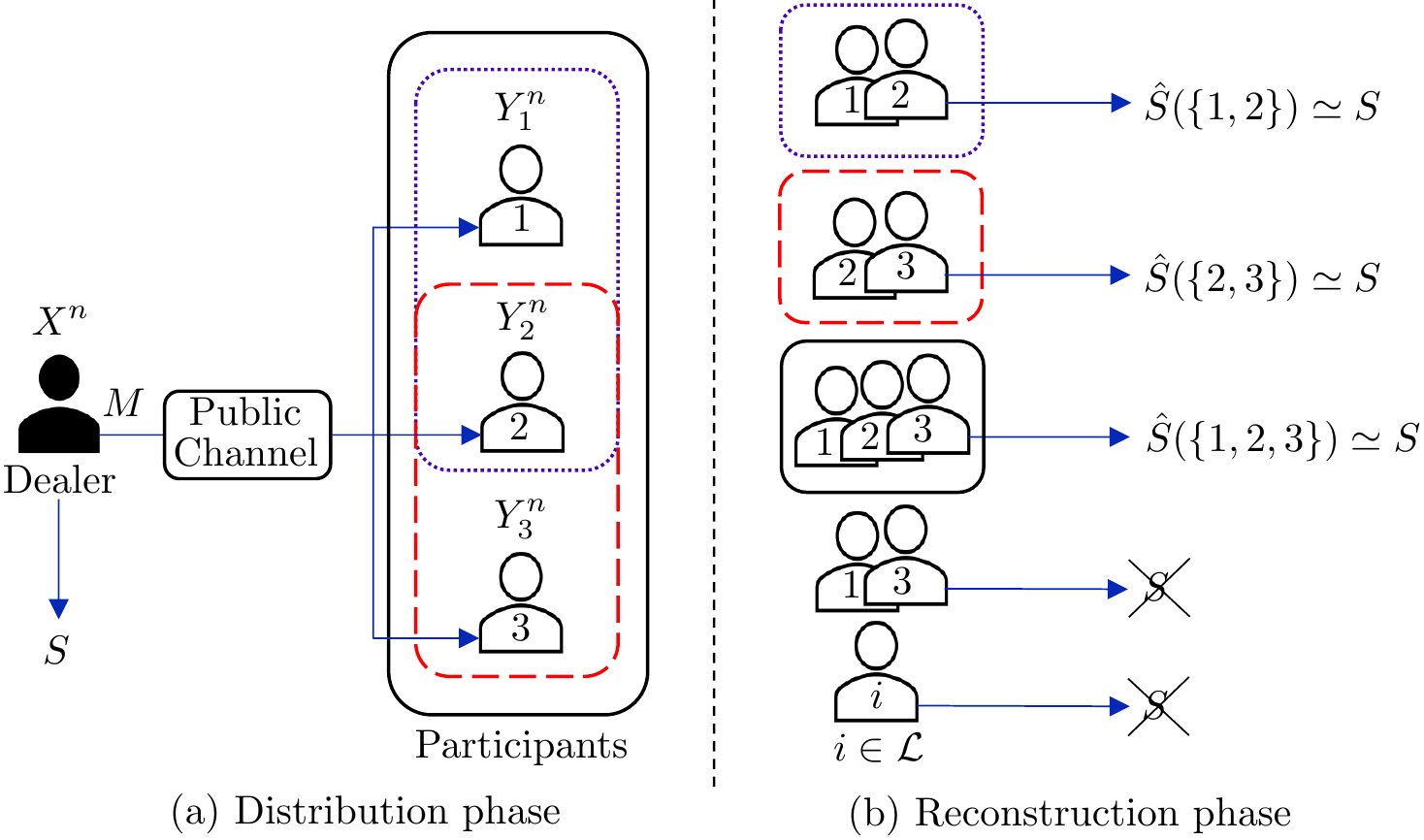}
\caption{Secret-sharing setting when $\mathbb{A}=\{\{1,2\}, \{2,3\}, \{1,2,3\}\}$ and $\mathbb{U}=\{ \{1,3\}, \{1\}, \{2\}, \{3\}\}$. Dashed, dotted, and solid contour lines represent the subsets of participants that are authorized to reconstruct the~secret.} \label{fig1}
\end{figure}

\begin{figure}
\centering
\begin{overpic}[trim=2.8cm 8.5cm 2.5cm 10cm,clip,width=.53\textwidth,tics=10]{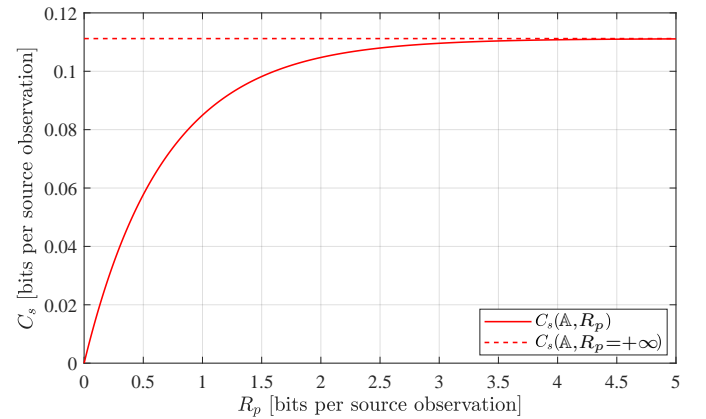}
\put (75.5,9.8) {${\scriptstyle\mathbb{A}, \scriptstyle R_p=+\infty}$}
\put (75.5,12.8) {$\scriptstyle\mathbb{A},\scriptstyle R_p$}
\end{overpic}
\caption{Secret capacity for Example \ref{ex1}.} \label{fig2}
\end{figure}
\subsection{Results for threshold access structures}

We now consider a special kind of access structure called threshold access structure \cite{shamir}.   
A threshold access structure with threshold $ t \in \llbracket 1, L \rrbracket$ is defined as $$\mathbb{A}_t\triangleq\{\mathcal{A}\subseteq \mathcal{L}:\vert \mathcal{A}\vert\geq t\}.$$ The complement of $\mathbb{A}_t$ is defined as  $\mathbb{U}_t\triangleq2^{\mathcal{L}}\backslash\mathbb{A}_t=\{\mathcal{A}\subseteq \mathcal{L}:\vert \mathcal{A}\vert < t\}.$ In other words, the threshold access structure is defined such that any set of $t$  participants can reconstruct the secret, but no set of fewer than $t$ participants can learn information about the secret.

The following result provides necessary and sufficient conditions to determine whether the secret capacity increases or decreases as the threshold $t$ increases. 

 \begin{theorem}\label{EX3}
 Assume that for any $\mathcal{S} \subseteq \mathcal{L}$, $H_\mathcal{S} = (H_\mathcal{L}(s))_{s \in \mathcal{S}}$. For any $t \in \llbracket 1 , L \rrbracket$, consider $\mathcal{A}_t^\star \in \argmin_{\mathcal{A}\in \mathbb{A}_{t}}H_\mathcal{A}^TH_\mathcal{A}$, and  $\mathcal{U}_t^\star \in \argmax_{\mathcal{U}\in \mathbb{U}_{t}} H_\mathcal{U}^TH_\mathcal{U}$. For any communication rate $R_p\geq 0$, for any $t \in \llbracket 1 , L \rrbracket$, we have
  \begin{align*}
     &C_s(\mathbb{A}_{1},R_p)\geq C_s(\mathbb{A}_{t},R_p),
     \end{align*}
and for any  $t \in \llbracket 1 , L \rrbracket$ and $i \in \llbracket 1 , L-t \rrbracket$,
  \begin{align*}
      &C_s(\mathbb{A}_{t},R_p) \geq C_s(\mathbb{A}_{t+i},R_p) \iff\\
      & \frac{ H^T_{\mathcal{U}^\star_{t+i}} H_{\mathcal{U}^\star_{t+i}} -  H^T_{\mathcal{U}^\star_{t}}H_{\mathcal{U}^\star_{t}} }{ H^T_{\mathcal{A}^\star_{t+i}}H_{\mathcal{A}^\star_{t+i}} -  H^T_{\mathcal{A}^\star_{t}}H_{\mathcal{A}^\star_{t}} }\geq  \frac{ 1+ \sigma^2_X  H^T_{\mathcal{U}_t^\star}H_{\mathcal{U}_t^\star}}{1+ \sigma^2_X  H^T_{\mathcal{A}_t^\star}H_{\mathcal{A}_t^\star}}.
  \end{align*}
  \end{theorem}

  \begin{proof}
See Appendix \ref{App_C}.
  \end{proof}
Theorem \ref{EX3} illustrates the fact that the secret capacity is not necessarily a monotonic decreasing function of the threshold~$t$.
 
  \begin{example}
Consider a dealer and five participants.  Assume that $\sigma^2_X\triangleq2$, $H_\mathcal{L}\triangleq[1, 0.85, 0.9, 0.95, 0.75]^T$,  and for any $\mathcal{S} \subseteq \mathcal{L}$, $H_\mathcal{S} = (H_\mathcal{L}(s))_{s \in \mathcal{S}}$. Then, one can compare the secret capacities for different thresholds using Theorem \ref{EX3}, as shown in Figure~\ref{fig3}. 
\begin{figure}[h]
\centering
\includegraphics[trim=2.5cm 8.5cm 3cm 9cm,clip,width=0.50\textwidth]{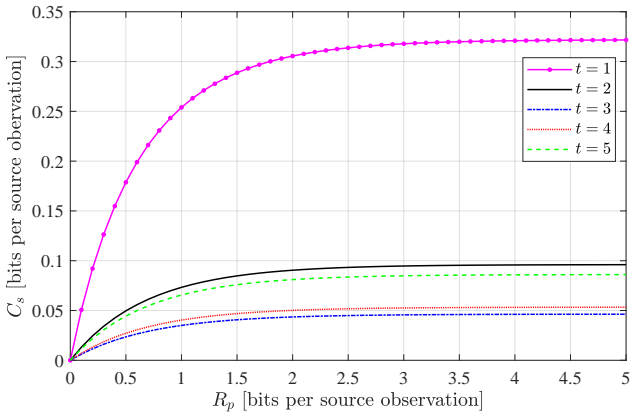}
\caption{Secret capacity for threshold access structure.} \label{fig3}
\end{figure}

From the definition of $\mathcal{A}_{t}^\star$ and $\mathcal{U}_{t}^\star$, we have 
$H_{\mathcal{A}^\star_1}=[0.75]^T$, 
$H_{\mathcal{A}^\star_2}=[0.75, 0.85]^T$,  $H_{\mathcal{A}^\star_3}=[0.75,0.85,0.9]^T$,
$H_{\mathcal{A}^\star_4}=[0.75,0.85,0.9,0.95]^T$,
$H_{\mathcal{A}^\star_5}=[0.75,0.85,0.9,0.95,1]^T$,
$H_{\mathcal{U}^\star_2}=[1]^T$,
$H_{\mathcal{U}^\star_3}=[1,0.95]^T$,  $H_{\mathcal{U}^\star_4}=[1,0.95,0.9]^T$,  and $H_{\mathcal{U}^\star_5}=[1,0.95,0.9,0.85]^T$.

For example,  
 putting $ H^T_{\mathcal{A}^\star_4}H_{\mathcal{A}^\star_4}=2.9975$,
$ H^T_{\mathcal{U}^\star_4}H_{\mathcal{U}^\star_4}=2.7125$,
$ H^T_{\mathcal{A}^\star_5} H_{\mathcal{A}^\star_5}=3.9975$, and
$ H^T_{\mathcal{U}^\star_5}H_{\mathcal{U}^\star_5}=3.4350$ in Theorem~\ref{EX3} with $t=4$ and $i=1$, we get 
 $C_s(\mathbb{A}_{4}, R_p)\leq C_s(\mathbb{A}_{5}, R_p)$ for any $R_p \geq 0$.
\end{example}

\section{Converse Proof of Theorem \ref{theo1}} \label{sec:conv}
To prove the converse, we first derive an upper bound on the secret capacity $C_s(\mathbb{A}, R_p)$ by considering a worst-case scenario in terms of a secret-key generation problem. This upper bound takes the form of a minimax optimization problem. We then derive a closed-form expression of this upper bound by proving a minimax theorem.

Define for $\mathcal{A} \in \mathbb{A}$, $\mathcal{U} \in \mathbb{U}$, $O_\mathcal{A}\triangleq H_\mathcal{A}^TH_\mathcal{A}$, and $O_\mathcal{U} \triangleq H_\mathcal{U}^TH_\mathcal{U}$.  Consider $V$ an auxiliary random variable jointly Gaussian with $X$, and let  $\sigma^2_{X\vert V}$ be the conditional variance of $X$ given $V$. Consider also $\mathcal{A}^\star\in\arg\min_{\mathcal{A} \in \mathbb{A}} O_\mathcal{A} $ and   $\mathcal{U}^\star\in\arg\max_{\mathcal{U} \in \mathbb{U}} O_\mathcal{U}$. Provided that $\sigma^2_{X\vert V} \neq 0$, for $\mathcal{A} \in \mathbb{A}$, $\mathcal{U} \in \mathbb{U}$, define 
\begin{align*}
    I_p (\sigma^2_{X \vert V}, \mathcal{A}) &\triangleq \frac{1}{2}\log\frac{\sigma^2_X}{\sigma^2_{X \vert V}}-\frac{1}{2}\log\frac{\sigma^2_XO_{\mathcal{A}}+1}{\sigma^2_{X \vert V}O_{\mathcal{A}}+1},\\
    I_s (\sigma^2_{X \vert V}, \mathcal{A}, \mathcal{U}) &\triangleq  \frac{1}{2}\log\frac{ \sigma^2_X O_\mathcal{A}+1}{\sigma^2_{X \vert V}O_\mathcal{A}+1}-\frac{1}{2}\log\frac{ \sigma^2_X O_\mathcal{U}+1}{\sigma^2_{X \vert V}O_\mathcal{U}+1}.
    \end{align*}
We will also use the following lemmas.
\begin{lemma}[Weinstein–Aronszajn identity, e.g., {\cite[Appendix~B]{sly}}] \label{Laug27_1}
 For any $\sigma^2 \in \mathbb{R}^+$ and  $A \in \mathbb{R}^{q\times 1}$, we have 
 \begin{align*}
 \det( A\sigma^2 A^T+I_q )=  A^T A\sigma^2+1.
 \end{align*}
 \end{lemma}
\begin{lemma}\label{Ldec6_1}
Let $c,d \in \mathbb{R}_+$ such that $c\geq d$. Then, the function $f_{c,d}$ is non-decreasing, where
\begin{align*}
f_{c,d}:\mathbb{R}_+&\rightarrow \mathbb{R} \nonumber \\
x & \mapsto \frac{1}{2}\log\frac{c x+1}{d x+1}.
\end{align*}
\begin{proof}
The derivative of $f_{c,d}$ at $x\in \mathbb{R}_+$ is
$f'_{c,d}(x)
=\frac{1}{2 \ln 2}\frac{c-d}{(cx+1)(dx+1)} \geq 0. $
\end{proof}
\end{lemma}
We now prove the converse of Theorem \ref{theo1} through a series of lemmas.
\begin{lemma} \label{lemconv1}
Let $R_p \in \mathbb{R}_+$. An upper bound on the secret capacity $C_s(\mathbb{A}, R_p)$  for the  Gaussian source model $(\mathcal{X}\times \mathcal{Y}_{\mathcal{L}},p_{XY_{\mathcal{L}}})$ is given by
\begin{align}
    C_s(\mathbb{A}, R_p) \leq  \min_{\mathcal{A} \in \mathbb{A}} \min_{\mathcal{U} \in \mathbb{U}} \max_{ \substack{ 0<\sigma^2_{X \vert V}\leq \sigma^2_X \\ \text{s.t. }  I_p (\sigma^2_{X \vert V}, \mathcal{A}) \leq R_p}} I_s (\sigma^2_{X \vert V}, \mathcal{A}, \mathcal{U}). \label{lemc}
\end{align}
\end{lemma}
\begin{proof}
Fix $\mathcal{A} \in \mathbb{A}$, $\mathcal{U} \in \mathbb{U}$. We first consider the secret-key generation model in \cite{r8} consisting of a transmitter (Alice), a receiver (Bob), and an eavesdropper (Eve), who observe $X^n$, $Y^n$, and $Z^n$, respectively, independently and identically distributed according to a Gaussian source $((\mathcal{X}\times\mathcal{Y}\times\mathcal{Z}), p_{XYZ})$, where $\mathcal{X}\triangleq \mathbb{R}$, $\mathcal{Y}\triangleq \mathbb{R}^{|\mathcal{A}|}$, $\mathcal{Z}\triangleq \mathbb{R}^{|\mathcal{U}|}$. In this model, a secret-key rate $R_k$ is achievable if after the transmission from Alice to Bob of message $M$ such that $H(M) \leq nR_p$ over an authenticated noiseless public channel, a secret key $K\in \llbracket 1, 2^{nR_k} \rrbracket$ is generated by Alice, and an estimate $\widehat{K}$ of $K$ is generated by Bob such that (i) $\lim_{n \rightarrow \infty} \mathbb{P}[K\neq \widehat{K}]=0$ (reliability), (ii) $\lim_{n \rightarrow \infty}  I(K;Z^n M)=0$ (security), 
and $\lim_{n \rightarrow \infty} \log\lceil2^{nR_k}\rceil-H(K)=0$ (uniformity). Moreover, the capacity region of this model is defined as  $\mathcal{R}({p_{XYZ}},\mathcal{A},\mathcal{U})\triangleq \{(R_p, R_k): (R_p,R_k) ~\text{is achievable}\}$.

Consider now the secret-sharing problem described in Section~\ref{secps} and the rate pair $(R_p,R_s) \in \mathcal{R}(p_{XY_{\mathcal{L}}},\mathbb{A})$. Then, by conditions (\ref{oct19_29_1}), (\ref{oct19_29_2}), and (\ref{oct19_29_3}), the rate pair $(R_p,R_s)$ also belongs to $\mathcal{R}({p_{XY_{\mathcal{A}}Y_{\mathcal{U}}}},\mathcal{A},\mathcal{U})$ for any $\mathcal{A} \in \mathbb{A}$, $\mathcal{U} \in \mathbb{U}$. Therefore, by \cite[Theorem 2]{r8}, we have for any $\mathcal{A} \in \mathbb{A}$, $\mathcal{U} \in \mathbb{U}$, 

\begin{align*}
    R_s & \leq \frac{1}{2}\log \left[ \frac{\det( H_\mathcal{A}\sigma^2_{X}H^T_\mathcal{A}+I)}{\det( H_\mathcal{A}\sigma^2_{X \vert V}H^T_\mathcal{A}+I)}   \frac{\det( H_\mathcal{U}\sigma^2_{X \vert V}H^T_\mathcal{U}+I)}{\det( H_\mathcal{U}\sigma^2_{X}H^T_\mathcal{U}+I)}\right]\!\!,\\
R_p & \geq \frac{1}{2}\log{ \frac{\sigma^2_{X}}{ \sigma^2_{X \vert V}} }-\frac{1}{2}\log \frac{\det( H_\mathcal{A}\sigma^2_{X}H_\mathcal{A}^T+I)}{\det( H_\mathcal{A} \sigma^2_{X \vert V}H_\mathcal{A}^T+I)},
\end{align*}
for some $\sigma^2_{X \vert V} \in ( 0, \sigma^2_X]$. 
Finally, using Lemma \ref{Laug27_1} and the definition of $O_\mathcal{A}$, $\mathcal{A} \in \mathbb{A}$ and $O_\mathcal{U}$, $\mathcal{U} \in \mathbb{U}$,  we have \eqref{lemc}.
  \end{proof}
  \begin{lemma} \label{lemconv3}
Let $R_p \in \mathbb{R}_+$. Let $\mathcal{A} \in \mathbb{A}$, $\mathcal{U} \in \mathbb{U}$, and assume that $O_{\mathcal{A}} \geq O_{\mathcal{U}}$. 
Then, we have
\begin{align}
    &\max_{ \substack{ 0<\sigma^2_{X \vert V}\leq \sigma^2_X \\ \text{s.t. }  I_p (\sigma^2_{X \vert V}, \mathcal{A}) \leq R_p}} I_s (\sigma^2_{X \vert V}, \mathcal{A}, \mathcal{U}) \nonumber\\ &=\frac{1}{2}\log\frac{\sigma^2_XO_{\mathcal{U}} 2^{-2R_p}+\sigma^2_X O_{\mathcal{A}} (1-2^{-2R_p})+1}{\sigma^2_{X}O_{\mathcal{U}}+1}. \label{eqoptl}
\end{align}
\end{lemma}

\begin{proof}
Fix $\mathcal{A} \in \mathbb{A}$ and $\mathcal{U} \in \mathbb{U}$. Let $\sigma^{2\star}_{X \vert V}(\mathcal{A},\mathcal{U})$ be an optimal solution on the left-hand side of~(\ref{eqoptl}). By writing $I_s (\sigma^2_{X \vert V}, \mathcal{A}, \mathcal{U})$ as
$$
I_s (\sigma^2_{X \vert V}, \mathcal{A}, \mathcal{U}) = \frac{1}{2}\log\frac{ \sigma^2_X O_\mathcal{A}+1}{ \sigma^2_X O_\mathcal{U}+1}-\frac{1}{2}\log\frac{\sigma^2_{X \vert V}O_\mathcal{A}+1}{\sigma^2_{X \vert V}O_\mathcal{U}+1},
$$
we have that $I_s (\sigma^2_{X \vert V}, \mathcal{A}, \mathcal{U})$ is a non-increasing function of $\sigma^2_{X \vert V}$  by Lemma \ref{Ldec6_1} because $O_{\mathcal{A}} \geq O_{\mathcal{U}}$. Hence, $\sigma^{2\star}_{X \vert V}(\mathcal{A},\mathcal{U})$ must be the smallest $\sigma^2_{X \vert V} \in ( 0, \sigma^2_X]$ that satisfies the constraint $I_p (\sigma^2_{X \vert V}, \mathcal{A}) \leq R_p$.  However, $I_p (\sigma^2_{X \vert V}, \mathcal{A})$ is a non-increasing function of $\sigma^2_{X \vert V}$; thus, we must have $I_p (\sigma^{2\star}_{X \vert V}(\mathcal{A},\mathcal{U}), \mathcal{A}) = R_p$, i.e., $$R_p=\frac{1}{2}\log\frac{\sigma^2_{X}}{\sigma^{2\star}_{X \vert V}(\mathcal{A},\mathcal{U})}-\frac{1}{2}\log\frac{\sigma^2_{X}O_{\mathcal{A}}+1}{\sigma^{2\star}_{X \vert V}(\mathcal{A},\mathcal{U})O_{\mathcal{A}}+1},$$
which gives
\begin{align} \label{eqoptsigm}
\sigma^{2\star}_{X \vert V}(\mathcal{A},\mathcal{U})=\frac{\sigma^2_{X}}{\sigma^2_{X}O_{\mathcal{A}}(2^{2R_p}-1) +2^{2R_p}}.
\end{align}
Plugging in this value for $\sigma^{2\star}_{X \vert V}(\mathcal{A},\mathcal{U})$ in $I_s (\sigma^{2\star}_{X \vert V}(\mathcal{A},\mathcal{U}), \mathcal{A}, \mathcal{U})$ gives \eqref{eqoptl}.

\end{proof}
\begin{lemma} \label{lemconv2}
Assume that for any $\mathcal{A} \in \mathbb{A}$, $\mathcal{U} \in \mathbb{U}$, we have $O_{\mathcal{A}} \geq O_{\mathcal{U}}$. Let $R_p \in \mathbb{R}_+$. Then, we have
\begin{align}
    &\min_{\mathcal{A} \in \mathbb{A}} \min_{\mathcal{U} \in \mathbb{U}} \max_{ \substack{ 0<\sigma^2_{X \vert V}\leq \sigma^2_X \\ \text{s.t. }  I_p (\sigma^2_{X \vert V}, \mathcal{A}) \leq R_p}} I_s (\sigma^2_{X \vert V}, \mathcal{A}, \mathcal{U})  \nonumber\\&= \max_{ \substack{ 0<\sigma^2_{X \vert V}\leq \sigma^2_X \\ \text{s.t. }  I_p (\sigma^2_{X \vert V}, \mathcal{A}^\star) \leq R_p}} \min_{\mathcal{A} \in \mathbb{A}} \min_{\mathcal{U} \in \mathbb{U}}  I_s (\sigma^2_{X \vert V}, \mathcal{A}, \mathcal{U}).\label{lemc2}
\end{align}
\end{lemma}  
\begin{proof}
By Lemma \ref{Ldec6_1}, we have for any $\sigma^2_{X \vert V} \in ( 0, \sigma^2_X]$, $\mathcal{A} \in \mathbb{A}$, $\mathcal{U} \in \mathbb{U}$, 
\begin{align*}
    \frac{1}{2}\log\frac{ \sigma^2_X O_\mathcal{A}+1}{\sigma^2_{X \vert V}O_\mathcal{A}+1} 
&\geq \frac{1}{2}\log\frac{ \sigma^2_X O_{\mathcal{A}^\star}+1}{\sigma^2_{X \vert V}O_{\mathcal{A}^\star}+1},\\
-\frac{1}{2}\log\frac{ \sigma^2_X O_\mathcal{U}+1}{\sigma^2_{X \vert V}O_\mathcal{U}+1} 
&\geq -\frac{1}{2}\log\frac{ \sigma^2_X O_{\mathcal{U}^\star}+1}{\sigma^2_{X \vert V}O_{\mathcal{U}^\star}+1};
\end{align*}
hence, $I_s({\sigma^2_{X \vert V},\mathcal{A},\mathcal{U}})\geq I_s({\sigma^2_{X\vert V},\mathcal{A}^\star,\mathcal{U}^\star})$, and we conclude that for any $\sigma^2_{X\vert V} \in ( 0, \sigma^2_X]$, 
\begin{align} 
    \min_{\mathcal{A} \in \mathbb{A}} \min_{\mathcal{U} \in \mathbb{U}} I_s (\sigma^2_{X \vert V}, \mathcal{A}, \mathcal{U})  =  I_s (\sigma^2_{X \vert V}, \mathcal{A}^\star, \mathcal{U}^\star).  \label{eqmax}
\end{align}
   Then, we have 
\begin{align*}
   &\min_{\mathcal{A} \in \mathbb{A}} \min_{\mathcal{U} \in \mathbb{U}} \max_{ \substack{ 0<\sigma^2_{X \vert V}\leq \sigma^2_X \\ \text{s.t. }  I_p (\sigma^2_{X \vert V}, \mathcal{A}) \leq R_p}} I_s (\sigma^2_{X \vert V}, \mathcal{A}, \mathcal{U}) \\
   & \stackrel{(a)}{=} \min_{\mathcal{A} \in \mathbb{A}} \min_{\mathcal{U} \in \mathbb{U}}  I_s (\sigma^{2\star}_{X \vert V}(\mathcal{A},\mathcal{U}), \mathcal{A}, \mathcal{U}) \\
        & \stackrel{(b)}{=}  I_s (\sigma_{X\vert V}^{2\star}(\mathcal{A}^\star,\mathcal{U}^\star), \mathcal{A}^\star , \mathcal{U}^\star) \\
         & =  \max_{ \substack{ 0<\sigma^2_{X \vert V}\leq \sigma^2_X \\ \text{s.t. }  I_p (\sigma^2_{X \vert V}, \mathcal{A}^\star) \leq R_p}}   I_s (\sigma^2_{X \vert V}, \mathcal{A}^\star, \mathcal{U}^\star) \\
    & \stackrel{(c)}{=}  \max_{ \substack{ 0<\sigma^2_{X \vert V}\leq \sigma^2_X \\ \text{s.t. }  I_p (\sigma^2_{X \vert V}, \mathcal{A}^\star) \leq R_p}} \min_{\mathcal{A} \in \mathbb{A}} \min_{\mathcal{U} \in \mathbb{U}}  I_s (\sigma^2_{X \vert V}, \mathcal{A}, \mathcal{U}) ,
\end{align*}
where in $(a)$ we have defined for $\mathcal{A} \in \mathbb{A}$, $\mathcal{U} \in \mathbb{U}$, $$\sigma^{2\star}_{X \vert V}(\mathcal{A},\mathcal{U}) \triangleq \displaystyle\argmax_{\substack{ 0<\sigma^2_{X \vert V}\leq \sigma^2_X \\ \text{s.t. }  I_p (\sigma^2_{X \vert V}, \mathcal{A}) \leq R_p}} I_s({\sigma^2_{X \vert V},\mathcal{A},\mathcal{U}}),$$  $(b)$ holds because for any $\mathcal{A} \in \mathbb{A}$, $\mathcal{U} \in \mathbb{U}$, we have  $I_s (\sigma^{2\star}_{X \vert V}(\mathcal{A},\mathcal{U}), \mathcal{A}, \mathcal{U}) \geq I_s (\sigma^{2\star}_{X \vert V}(\mathcal{A},\mathcal{U}), \mathcal{A}
^\star, \mathcal{U}^\star) \geq I_s (\sigma^{2\star}_{X \vert V}(\mathcal{A}^\star,\mathcal{U}^\star), \mathcal{A}
^\star, \mathcal{U}^\star)$, where the first inequality holds by \eqref{eqmax}, and the second inequality holds because $I_s (\sigma^{2\star}_{X \vert V}(\mathcal{A},\mathcal{U}), \mathcal{A}
^\star, \mathcal{U}^\star)$ is a non-increasing function of $\sigma^{2\star}_{X \vert V}(\mathcal{A},\mathcal{U})$ by Lemma~\ref{Ldec6_1}, and $\sigma^{2\star}_{X \vert V}(\mathcal{A}^\star,\mathcal{U}^\star) \geq \sigma^{2\star}_{X \vert V}(\mathcal{A},\mathcal{U})$ by \eqref{eqoptsigm} in the proof of Lemma \ref{lemconv3}, and $(c)$ holds by \eqref{eqmax}.
\end{proof}

Next, we remark that if there exist $\mathcal{A} \in \mathbb{A}$ and $\mathcal{U} \in \mathbb{U}$ such that $O_{\mathcal{A}} < O_{\mathcal{U}}$, then $C_s(\mathbb{A}, R_p)=0$ by Lemma \ref{lemconv1} and Lemma \ref{Ldec6_1} applied to $f_{\sigma^{2}_{X},\sigma^{2}_{X \vert V}}$. Thus, we obtain the converse of Theorem \ref{theo1} by combining Lemmas \ref{lemconv1}, \ref{lemconv3}, and \ref{lemconv2}.
\section{Achievability Proof of Theorem \ref{theo1}} \label{sec:ach}

To prove the achievability part of Theorem \ref{theo1}, we first prove an achievability result for discrete random variables in Section~\ref{secdis} and then extend our result to Gaussian random variables by a quantization argument in Section \ref{seccont}.

\subsection{Discrete case} \label{secdis}
Our coding scheme decouples the requirements \eqref{oct19_29_1} (reliability) and \eqref{oct19_29_2} (security with respect to unauthorized groups of colluding users). Specifically, as described next, we repeat $q \in \mathbb{N}$ times a reconciliation step to handle \eqref{oct19_29_1} via a compound version of Wyner-Ziv coding and then perform a privacy amplification step to handle~\eqref{oct19_29_2} via  universal hashing implemented with extractors. Note that Wyner-Ziv coding is a key component to handle rate-limited communication constraints as in rate-limited secret-key generation \cite{csiszar2000} and biometric secrecy system models, e.g., \cite{chou2015polar,igna2012,igna2013,gunlu2021,gunlu2019}, which rely on rate-limited secret-key generation. Here, unlike in \cite{chou2015polar,igna2012,igna2013,gunlu2021,gunlu2019}, we employ a compound version of Wyner-Ziv coding because unlike in \cite{chou2015polar,igna2012,igna2013,gunlu2021,gunlu2019}, we simultaneously consider multiple reliability constraints due to the presence of an access structure.
\subsubsection{Reconciliation step}\label{secrec}
Let $n \in \mathbb{N}$ and $\epsilon>0$.  For a probability mass function $p_X$, denote the set of $\epsilon$-letter typical sequences~\cite{orlitsky2001coding} (see also \cite{r4}) with respect to $p_X$ by $\mathcal{T}_{\epsilon}^n(X)$, and define $supp(p_X)\triangleq \{x \in \mathcal{X}:p_X(x)>0\}$ and $\mu_X\triangleq\min_{x \in supp(p_X)}p_X(x)$. Define $\epsilon_1\triangleq\frac{1}{2}\epsilon$. \\
{\bf Code construction:} Fix a joint probability distribution $p_{VXY_{\mathcal{L}}}$ on $\mathcal{V}\times \mathcal{X}\times \mathcal{Y}_{\mathcal{L} }$,  where $V$ is an auxiliary random variable such that $ V - X -Y_{\mathcal{L}} $ forms a Markov chain. Define $R_v\triangleq  \max_{\mathcal{A} \in {\mathbb{A}}}H(V\vert Y_{\mathcal{A} }) - H(V|X)+6\epsilon H(V)$, $R_v'\triangleq H(V) - \max_{\mathcal{A} \in {\mathbb{A}}} H(V|Y_{\mathcal{A} })-3\epsilon H(V)$. Generate $2^{n(R_v+R_v')}$ codewords, labeled $v^n(\omega,\nu) $ with $ (\omega,\nu)\in \llbracket 1,2^{nR_v} \rrbracket \times \llbracket 1,2^{nR_v'} \rrbracket$, by generating the symbols $v_i(\omega,\nu)$ for $i \in \llbracket 1,n \rrbracket$ and $(\omega,\nu)\in \llbracket 1,2^{nR_v} \rrbracket \times \llbracket 1,2^{nR_v'} \rrbracket$ independently according to~$p_V$.\\
{\bf Encoding:} Given $x^n$, find a pair $(\omega,\nu)$ such that $(x^n,v^n(\omega,\nu))\in \mathcal{T}_{\epsilon}^n(XV)$. If there are several pairs, choose one (according to the lexicographic order); otherwise, set $(\omega, \nu)=(1,1)$. Define $v^n\triangleq v^n(\omega,\nu)$, and transmit $m \triangleq \omega$.\\
{\bf Decoding:} Let $\mathcal{A} \in \mathbb{A}$. Given $y_{\mathcal{A} }^n$ and $m$, find $\tilde{\nu}_{\mathcal{A} }$ such that $(y_{\mathcal{A} }^n,v^n(\omega,\tilde{\nu}_{\mathcal{A} }))\in \mathcal{T}_{\epsilon}^n(Y_{\mathcal{A} }V)$. If there is one or more  $\tilde{\nu}_{\mathcal{A} }$, then choose the smallest; otherwise, set  $\tilde{\nu}_{\mathcal{A} }=1$. Define $\widehat{v}_{\mathcal{A}}^n\triangleq v^n(\omega,\tilde{\nu}_{\mathcal{A}})$.\\
\noindent{}{\bf Probability of error:}
The random variable that represents the randomly generated code is denoted by $C_n$.  As shown in Appendix \ref{Appeq10}, there exists a codebook $\mathcal{C}_n^\star$ such that 
\begin{align} \label{eqerror}
\max_{\mathcal{A} \in \mathbb{A}} \mathbb{P}[V^n\neq \widehat{V}^n_\mathcal{A}] \leq \vert\mathbb{A} \vert\max_{\mathcal{A} \in \mathbb{A}}\delta(n,\epsilon,\mathcal{A}),
\end{align}
where 
\begin{align*}
    \delta(n,\epsilon,\mathcal{A}) 
    &\triangleq 2 \vert \mathcal{X}\vert\vert\mathcal{Y}_\mathcal{A}\vert e^{-n\epsilon_1^2\mu_{XY_\mathcal{A}}}+2^{-n\epsilon H(V)}\\
    &\phantom{-} + \exp (-(1-2\vert\mathcal{V}\vert \vert\mathcal{X} \vert e^{-n\frac{(\epsilon-\epsilon_1)^2}{1+\epsilon_1}\mu_{VX}})2^{\epsilon n H(V)})\\
    &\phantom{-} +2 \vert \mathcal{V}\vert\vert \mathcal{X}\vert\vert\mathcal{Y}_\mathcal{A}\vert e^{-n\frac{(\epsilon-\epsilon_1)^2}{1+\epsilon_1}\mu_{VXY_\mathcal{A}}}.
\end{align*}

\subsubsection{Privacy amplification step}
Let $q, n \in \mathbb{N}$, and define $N\triangleq nq$.   The reconciliation step is repeated $q$ times such that the dealer has $V^N = (V^n)^q$ and the participants in $\mathcal{A}\in \mathbb{A}$ have $(\widehat{V}^n_\mathcal{A})^q$. Note that the total public communication $M \in \mathcal{M}$ is such that $\frac{ H(M)}{N} \leq  \frac{\log |\mathcal{M}|}{N} = \max_{\mathcal{A}\in \mathbb{A}}I(X;V\vert Y_{\mathcal{A} })+6\epsilon H(V)$.
Next, another round of reconciliation with negligible communication is performed to ensure that $\max_{\mathcal{A}\in \mathbb{A}} \mathbb{P}[(V^n)^q\neq (\widehat{V}^n_\mathcal{A})^q] \leq \delta(q)$, where $\lim_{q \to \infty } \delta(q) = 0$ when $n$ is fixed.
Finally, the dealer computes $S=g(V^N,U_d)$, while the participants in $\mathcal{A}\in \mathbb{A}$ compute $\widehat{S}({\mathcal{A} })=g(\widehat{V}_{\mathcal{A} }^N,U_d)$, where $U_d$ is a sequence of $d$ (to be defined later) uniformly distributed random bits, and $g:\{0, 1\}^N\times \{0, 1\}^d \rightarrow \{0, 1\}^k$ is to be defined later.

\subsubsection*{Analysis of reliability}

 The secrets computed by the dealer and the participants in $\mathcal{A}\in \mathbb{A}$ are asymptotically the same for a fixed $n$ as $q$ goes to infinity.
\begin{align*}
\mathbb{P}[\widehat{S}(\mathcal{A})\neq S]\leq \mathbb{P}[ (\widehat{V}^n_\mathcal{A})^q \neq (V^n)^q]\leq \delta(q).
\end{align*}
\subsubsection*{Analysis of security}
Let the min-entropy of a discrete random variable $X$, defined over $\mathcal{X}$ with probability mass function  $p_X$, be denoted by $H_{\infty}(X) \triangleq - \log \left( \max_{x\in \mathcal{X}} p_X(x) \right)$. We will use the following lemmas: 
\begin{lemma}[Adapted from \cite{maurer}]\label{Ldec6_2}
Let $E_{\mathcal{U}}$ be the random variable that represents the total knowledge about $V^N$ available to participants in $\mathcal{U} \in \mathbb{U}$. Let $e_{\mathcal{U}}$ be a particular realization of $E_{\mathcal{U}}$. 
If $H_\infty(V^N\vert E_{\mathcal{U}}=e_{\mathcal{U}})\geq\gamma N$, for some $\gamma\in [0, 1] \backslash \{0,1\}$, then there exists an extractor $g:\{0, 1\}^N\times \{0, 1\}^d \rightarrow \{0, 1\}^k$ with $d \leq N \delta(N)$ and $k\geq N (\gamma-\delta(N))$, where $\delta(N)$ is such that $\lim_{N\rightarrow+\infty}\delta(N)=0$. Moreover, 
\begin{align*}
H(S\vert U_d, &E_{\mathcal{U}}=e_{\mathcal{U}})\geq k-\delta^\star(N),
\end{align*}
with $\delta^\star(N)=2^{-\sqrt{N}/\log N}(k+\sqrt{N}/\log N)$.
\end{lemma}
\begin{lemma}[\!\cite{maurer}, see also  \cite{r2}]{\label{refversion}}
Consider a discrete memoryless source  $(\mathcal{X}\times \mathcal{Y},p_{XY})$ and define  \begin{align*}
\Theta \triangleq \mathbbm{1}\{(X^q,Y^q)\in \mathcal{T}^q_{2\epsilon}(XY) \} \mathbbm{1}\{Y^q\in \mathcal{T}^q_{\epsilon}(Y) \}. 
\end{align*}
Then, $\mathbb{P}[\Theta=1]\geq 1-(2\vert S_{X}\vert e^{-\epsilon^2q\mu_{X}/3}+2\vert S_{XY}\vert e^{-\epsilon^2q\mu_{XY}/3} )$, with $ S_{XY} \triangleq supp(  p_{XY } )$ and $ S_{Y } \triangleq supp(  p_{Y} )$. Moreover, if $y^q \in \mathcal{T}^q_{\epsilon}(Y)$, then
\begin{align*}
&H_{\infty}(X^q\vert Y^q =y^q , \Theta =1) \\&\geq q (1- \epsilon) H(X|Y) + \log( 1 - 2\vert S_{XY}\vert e^{-\epsilon^2q\mu_{XY}/6}).
\end{align*}
\end{lemma}
Define for any $\mathcal{U} \in \mathbb{U}$, the random variables
\begin{align}
    \Theta_\mathcal{U}  &\triangleq  \mathbbm{1}\{(V^N,Y_\mathcal{U}^N) \in \mathcal{T}_{2\epsilon}^q(V^nY^n_\mathcal{U}) \}\mathbbm{1}\{Y_\mathcal{U}^N \in  \mathcal{T}_{\epsilon}^q(Y^n_\mathcal{U})\},\\
\Upsilon_\mathcal{U} &\triangleq \mathbbm{1} \{H_\infty(V^N\vert Y_\mathcal{U}^N=y_\mathcal{U}^N,\Theta_\mathcal{U}=1)\nonumber\\
&\phantom{---}-H_\infty(V^N\vert Y_\mathcal{U}^N=y_\mathcal{U}^N, M=m,\Theta_\mathcal{U}=1) \nonumber \\
&\phantom{---}\leq \log \vert \mathcal{M} \vert + \sqrt{N}\}. \label{equpsilon1}
\end{align}
For any $\mathcal{U} \in \mathbb{U}$, $\mathbb{P}[\Theta_\mathcal{U} =1 ] \geq 1-\delta^0_{\epsilon}(n,\mathcal{U})$, where $\delta^0_{\epsilon}(n,\mathcal{U}) \triangleq 2\vert S_{V^n}\vert e^{-\epsilon^2q\mu_{V^n}/3}+2\vert S_{V^nY_{\mathcal{U}}^n}\vert e^{-\epsilon^2q\mu_{V^nY_{\mathcal{U}}^n}/3} $ by Lemma~\ref{refversion} applied to the discrete memoryless source model $(\mathcal{V}^n\times\mathcal{Y}^n_\mathcal{U}, p_{V^nY^n_\mathcal{U}})$, and $\mathbb{P}[\Upsilon_\mathcal{U} =1 ] \geq 1-2^{-\sqrt{N}}$ by \cite[Lemma 10]{maurer}. 
Hence,
\begin{align}
\mathbb{P}[\Upsilon_\mathcal{U} =1 , \Theta_\mathcal{U}=1] \geq 1-\delta^0_{\epsilon}(n,\mathcal{U})-2^{-\sqrt{N}}. \label{eqprobupsilon}
\end{align}
Then, for any $\mathcal{U}\in \mathbb{U}$, we have
\begin{align}
H(S \vert U_d Y_\mathcal{U}^N M)
& \stackrel{(a)}\geq H(S \vert U_d Y^N_{\mathcal{U}} M \Theta_{\mathcal{U}} \Upsilon_{\mathcal{U}} ) \nonumber \\
&\geq\min_{\mathcal{U} \in \mathbb{U}} H(S \vert U_d Y_\mathcal{U}^N M \Theta_\mathcal{U} \Upsilon_\mathcal{U} ) \nonumber \\
&\geq  \min_{\mathcal{U} \in \mathbb{U}} \mathbb{P}[\Theta_\mathcal{U}=1, \Upsilon_\mathcal{U} =1] \nonumber\\&\hspace{1.2cm}\times H(S \vert U_d Y_\mathcal{U}^N M,\Theta_\mathcal{U}=1, \Upsilon_\mathcal{U} =1) \nonumber\\
&\geq \min_{\mathcal{U} \in \mathbb{U}} \mathbb{P}[\Theta_\mathcal{U}=1, \Upsilon_\mathcal{U} =1]\nonumber\\&\hspace{0.6cm} \times\min_{\mathcal{U} \in \mathbb{U}} H(S \vert U_d Y_\mathcal{U}^N M ,\Theta_\mathcal{U}=1, \Upsilon_\mathcal{U} =1 )\nonumber\\
&\stackrel{(b)}  \geq \left(1-\max_{\mathcal{U} \in \mathbb{U}}\delta^0_{\epsilon}(n,\mathcal{U})-2^{-\sqrt{N}}\right)\nonumber\\&  \hspace{0.5cm}\times\min_{\mathcal{U} \in \mathbb{U}}H(S \vert U_d Y_\mathcal{U}^N M,\Theta_\mathcal{U}=1, \Upsilon_\mathcal{U} =1 ),
\label{rl1}
\end{align}
where $(a)$ holds because conditioning reduces entropy and $(b)$ holds by \eqref{eqprobupsilon}.
 To lower bound $\min_{\mathcal{U} \in \mathbb{U}}H(S \vert  U_d Y_\mathcal{U}^N M,\Theta_\mathcal{U}=1, \Upsilon_\mathcal{U} =1)$ in (\ref{rl1}) with Lemma \ref{Ldec6_2}, we now lower bound
$\min_{\mathcal{U} \in \mathbb{U}}H_\infty(V^N \vert  Y_\mathcal{U}^N=y_\mathcal{U}^N,M=m, \Theta_\mathcal{U}=1, \Upsilon_\mathcal{U}=1)$. We have for any $\mathcal{U} \in \mathbb{U}$,
\begin{align}
& H_\infty(V^N\vert Y_\mathcal{U}^N=y_\mathcal{U}^N,M=m, \Theta_\mathcal{U}=1, \Upsilon_\mathcal{U}=1)  \nonumber\\
& \stackrel{(a)}\geq H_\infty(V^N\vert Y_\mathcal{U}^N=y_\mathcal{U}^N, \Theta_\mathcal{U}=1)-\log \vert \mathcal{M} \vert-\sqrt{N} \nonumber\\
& \stackrel{(b)}{\geq} q (1- \epsilon){H}(V^n\vert Y_\mathcal{U}^n) -\delta^1_{\epsilon}(q,n,\mathcal{U}) -N(\max_{\mathcal{A} \in \mathbb{A}}I(V;X\vert Y_\mathcal{A})\nonumber\\&\hspace{5.7cm}+ 6 \epsilon H(V))-\sqrt{N} \nonumber\\
& \stackrel{(c)}{\geq}  N [I(X;V\vert Y_\mathcal{U}) -\max_{\mathcal{A} \in \mathbb{A}}I(V;X\vert Y_\mathcal{A}) - \delta^2_{\epsilon}(q,n,\mathcal{U}) ] \nonumber \\
&\geq  N [ \min_{\mathcal{U} \in \mathbb{U}}I(X;V\vert Y_\mathcal{U}) \! -\max_{\mathcal{A} \in \mathbb{A}}I(V;X\vert Y_\mathcal{A}) \! - \max_{\mathcal{U} \in \mathbb{U}}\delta^2_{\epsilon}(q,n,\mathcal{U}) ] \nonumber \\
& \stackrel{(d)} = N[ \min_{\mathcal{A} \in \mathbb{A}}I(V;Y_\mathcal{A})-\max_{\mathcal{U} \in \mathbb{U}}I(V;Y_\mathcal{U}) - \max_{\mathcal{U} \in \mathbb{U}}\delta^2_{\epsilon}(q,n,\mathcal{U})],
\label{rl7}
\end{align}
where $(a)$ holds by \eqref{equpsilon1}, $(b)$ holds by Lemma \ref{refversion} with $\delta^1_{\epsilon}(q,n,\mathcal{U}) \triangleq - \log( 1 - 2\vert S_{V^nY_{\mathcal{U}}^n}\vert e^{-\epsilon^2q\mu_{V^nY_{\mathcal{U}}^n}/6})$, $(c)$ holds with  $\delta^2_{\epsilon}(q,n,\mathcal{U}) \triangleq  \epsilon  I(X;V\vert Y_\mathcal{U}) + (1-\epsilon) [ 2\epsilon  H(X\vert Y_\mathcal{U}V)+2n^{-1} + \log\vert\mathcal{X}\vert(4\vert \mathcal{V} \vert |\mathcal{X}| e
^{-n\epsilon^2\mu_{XV}} + 2\vert \mathcal{V} \vert |\mathcal{X}| \vert \mathcal{Y}_{\mathcal{U}}\vert e^{-\epsilon^2n\mu_{VXY_\mathcal{U}}/8})]  + N^{-1}\delta^1_{\epsilon}(q,n,\mathcal{U})+  6  \epsilon H(V) +N^{-1/2} $ because, as shown in Appendix \ref{App16}, we have 
\begin{align}
H({V^n\vert Y^n_\mathcal{U}}) 
\geq n(H(X\vert Y_\mathcal{U})-H(X\vert Y_\mathcal{U}V)(1+2\epsilon))\nonumber \\
-2-n\log\vert\mathcal{X}\vert(4\vert \mathcal{V} \vert |\mathcal{X}| e^{-n\epsilon^2\mu_{XV}}\nonumber \\+ 2\vert \mathcal{V} \vert |\mathcal{X}| \vert \mathcal{Y}_{\mathcal{U}}\vert e^{-\epsilon^2n\mu_{VXY_\mathcal{U}}/8}),
\label{rl6}
\end{align}
and $(d)$ holds because $V \mathdash X \mathdash (Y_{\mathcal{A} }, Y_{\mathcal{U} })$.

Next, we set the output size $k$ of the extractor to be less than the lower bound in (\ref{rl7}) by $\sqrt{N}$, i.e.,
\begin{align}
 k & \triangleq \lfloor N[ \min_{\mathcal{A} \in \mathbb{A}}I(V;Y_\mathcal{A})-\max_{\mathcal{U} \in \mathbb{U}}I(V;Y_\mathcal{U}) - \max_{\mathcal{U} \in \mathbb{U}}\delta^2_{\epsilon}(q,n,\mathcal{U})\nonumber\\& \phantom{--}- N^{-1/2}] \rfloor,
\label{rl9}
\end{align}
Finally, we have 
\begin{align}
\max_{\mathcal{U} \in \mathbb{U}}I(S;U_dY_\mathcal{U}^NM)&=H(S)-\min_{\mathcal{U} \in \mathbb{U}}H(S\vert U_dY_\mathcal{U}^NM) \nonumber \\
& \stackrel{(a)}{\leq}k-\left(1-\max_{\mathcal{U} \in \mathbb{U}}\delta^0_{\epsilon}(n,\mathcal{U})-2^{-\sqrt{N}}\right)\nonumber\\&\hspace{3.2cm}\times\left(k- \delta^\star(N)\right)\nonumber\\
&\stackrel{(b)} \leq  \delta_{\epsilon}^3(N), \label{nov11_14_1}
\end{align}
where $(a)$ holds by \eqref{rl1}, \eqref{rl7} (valid for any $\mathcal{U} \in \mathbb{U}$), \eqref{rl9}, and Lemma \ref{Ldec6_2} with $\delta^\star(N)\triangleq2^{-\sqrt{N}/\log N}\left(k+\sqrt{N}/\log N\right)$ and $(b)$ holds with $\delta_{\epsilon}^3(N)\triangleq \delta^\star(N)+\left(\max_{\mathcal{U} \in \mathbb{U}}\delta^0_{\epsilon}(n,\mathcal{U})+2^{-\sqrt{N}}\right)k$. 
\subsubsection*{Analysis of uniformity}
 Similar to (\ref{nov11_14_1}), we have 
\begin{align}
~~~{H(S)} &\geq \min_{\mathcal{U} \in \mathbb{U}}H(S \vert U_dY_\mathcal{U}^NM) \nonumber \\
&\geq k-\delta_{\epsilon}^3(N). 
\label{c1eeee}
\end{align}
\subsubsection*{Public communication rate}
The public communication rate corresponds to the rate of $M$ plus the rate of $U_d$, i.e.,
\begin{align*}
\lim_{N\to\infty} R_p
&=\max_{\mathcal{A}\in \mathbb{A}}I(X;V\vert Y_{\mathcal{A} })+6\epsilon H(V) .
\end{align*}
\subsubsection*{Achievable secret rate}
The secret rate $R_s\triangleq k/N$ satisfies 
\begin{align}
 R_s &\geq  \min_{\mathcal{A} \in \mathbb{A}}I(V;Y_\mathcal{A})-\max_{\mathcal{U} \in \mathbb{U}}I(V;Y_\mathcal{U}) - \max_{\mathcal{U} \in \mathbb{U}}\delta^2_{\epsilon}(q,n,\mathcal{U})\nonumber\\ & \hspace{4.8cm} - N^{-1/2} -N^{-1}.
\label{rl10}
\end{align}
\subsection{Continuous case} \label{seccont}
In this section, we extend the achievability result of Section \ref{secdis} for discrete random variables  to Gaussian random variables by means of quantization. Quantization also allows us to  show that the size of the shares linearly scales with the length of the secret.  The main issue with quantization is that it might lead to an underestimation of the information that unauthorized sets of participants may learn about the secret. We will, however, show that this issue can be overcome provided that the quantization is fine enough.

We now build upon Section \ref{secdis} to show that $(R_p,R_s) \in \mathcal{R}(p_{XY_{\mathcal{L}}}, \mathbb{A})$, where 
\begin{align}
R_p&=\frac{1}{2}\log\frac{\sigma^2_{X}}{\sigma^2_{X \vert V}}-\frac{1}{2}\log\frac{\sigma^2_{X}O_{\mathcal{A}^\star}+1}{\sigma^2_{X \vert V}O_{\mathcal{A}^\star}+1}, \label{eqf1} \\
R_s &= \frac{1}{2}\log\frac{ \sigma^2_{X} O_{\mathcal{A}^\star}+1}{\sigma^2_{X \vert V}O_{\mathcal{A}^\star}+1}-\frac{1}{2}\log\frac{ \sigma^2_{X} O_{\mathcal{U}^\star}+1}{\sigma^2_{X \vert V}O_{\mathcal{U}^\star}+1} . \label{eqf2}
\end{align}
We use the following lemma to extend Section \ref{secdis} to the continuous case by means of quantization.

\begin{lemma}[\!\!\cite{r5,r6,r7}]\label{33}   Let $X$ and $Y$ be two real-valued random variables with probability distribution $\mathbb{P}_X$ and $\mathbb{P}_{Y}$, respectively. Let $\mathcal{C}_{\Delta_1}=\{ C_i \}_{i \in \mathcal{I}}$, $\mathcal{D}_{\Delta_2}=\{ D_j \}_{j \in \mathcal{J}}$ be two partitions of the real line for $X$ and $Y$ such that for any $i \in \mathcal{I}$, $\mathbb{P}_X[C_i]=\Delta_1$, for any $j \in \mathcal{J}$, $\mathbb{P}_{Y}[D_j]=\Delta_2$, where $\Delta_1, \Delta_2 > 0$. Let $X_{\Delta_1}, Y_{{{\Delta_2}}}$ be the quantized version of $X, Y$ with respect to the partitions $\mathcal{C}_{\Delta_1}, \mathcal{D}_{\Delta_2}$, respectively. Then, we~have 
$$I(X,Y)=\lim_{\Delta_1, \Delta_2 \rightarrow 0} I(X_{\Delta_1}, Y_{{ {\Delta_2}}}).$$
\end{lemma}
We first show that a quantization does not affect the security requirement \eqref{oct19_29_2}.
\begin{proposition}
A quantization of  $Y_\mathcal{U}^n$, $\mathcal{U} \in \mathbb{U}$, might lead to an underestimation of $I(S;M,Y^n_\mathcal{U})$. 
However, if the quantized version $Y^n_{{\mathcal{U}},{\Delta}}$ of $Y_\mathcal{U}^n$, $\mathcal{U} \in \mathbb{U}$, is fine enough, then for any $\delta>0$
\begin{align} \label{eqquantleakl}
    \max_{\mathcal{U} \in \mathbb{U}} I(S;MY_\mathcal{U}^n) \leq \max_{\mathcal{U} \in \mathbb{U}}I(S;M Y^n_{{\mathcal{U}},{\Delta}}) + \delta.
\end{align}
\end{proposition}
\begin{proof} 
For any $\delta > 0$, for any $\mathcal{U} \in \mathbb{U}$, we have
\begin{align}
I(S;MY_\mathcal{U}^n) 
& \leq \vert I(S;MY_\mathcal{U}^n)-I(S;M Y^n_{{\mathcal{U}},{\Delta}}) \vert + I(S;M Y^n_{{\mathcal{U}},{\Delta}}) \nonumber \\
& \leq \max_{\mathcal{U} \in \mathbb{U}} \vert I(S;MY_\mathcal{U}^n)-I(S;M Y^n_{{\mathcal{U}},{\Delta}}) \vert\nonumber\\&\hspace{4cm} + \max_{\mathcal{U} \in \mathbb{U}}  I(S;M Y^n_{{\mathcal{U}},{\Delta}}) \nonumber \\
&\leq  \delta + \max_{\mathcal{U} \in \mathbb{U}}  I(S;M Y^n_{{\mathcal{U}},{\Delta}}), \label{eqquantleak}
\end{align}
where the last inequality holds by Lemma \ref{33}, if the quantized version $Y^n_{{\mathcal{U}},{\Delta}}$ of $Y_\mathcal{U}^n$, $\mathcal{U} \in \mathbb{U}$, is fine enough. Since \eqref{eqquantleak} is valid for any $\mathcal{U} \in \mathbb{U}$, we obtain \eqref{eqquantleakl}.
\end{proof}
For $\mathcal{A} \in \mathbb{A}$ and $\mathcal{U} \in \mathbb{U}$, we quantize $X,Y_\mathcal{A},Y_\mathcal{U},~\text{and}~V$ as in Lemma~\ref{33} to form $X_{\Delta},Y_{\mathcal{A},\Delta},Y_{\mathcal{U},\Delta}$, and $V_{\Delta}$ such that $\Delta=l^{-1}$ and $\vert \mathcal{X}_{\Delta} \vert=\vert \mathcal{Y}_{\mathcal{A},\Delta} \vert=\vert \mathcal{Y}_{\mathcal{U},\Delta}\vert=\vert \mathcal{V}_{\Delta} \vert=l $ with $l >0$.  Next, we apply the proof for the discrete case to the random variables $X_{\Delta},Y_{\mathcal{A},\Delta},Y_{\mathcal{U},\Delta}$, $V_{\Delta}$. By Lemma \ref{33}, we can fix $l$ large enough such that, for any $\mathcal{A} \in \mathbb{A}$, $\vert I(V_{\Delta};Y_{\mathcal{A},\Delta})-I(V;Y_\mathcal{A}) \vert<\delta/2$, for any $\mathcal{U} \in \mathbb{U}$, $\vert I(V_{\Delta};Y_{\mathcal{U},\Delta})-I(V;Y_\mathcal{U}) \vert<\delta/2 $, such that~(\ref{rl10}) becomes
\begin{align*}
&R_s \geq  \min_{\mathcal{A} \in \mathbb{A}}I(V;Y_\mathcal{A})-\max_{\mathcal{U} \in \mathbb{U}}I(V;Y_\mathcal{U}) - \max_{\mathcal{U} \in \mathbb{U}}\delta^2_{\epsilon}(q,n,\mathcal{U})\\&\hspace{5cm}- N^{-1/2} -N^{-1} - \delta.
\end{align*}
Note that $\delta^2_{\epsilon}(q,n,\mathcal{U})$, $\mathcal{U} \in \mathbb{U}$, in the above equation hides the terms  $2\epsilon (1-\epsilon) H(X_{\Delta}\vert Y_{\mathcal{U},\Delta}V_{{\Delta}}) $ and $6 \epsilon H(V_{\Delta})$,  which do not go to zero as $l$ goes to infinity. Consequently, we choose $\epsilon=n^{-\alpha}$, where $\alpha \in [0, 1/2] \backslash \{0,1/2\} $, such that 
if we choose~$l$ large enough, then $n$ large enough, and finally $q$ large enough, then the asymptotic secret rate is as close as desired to 
\begin{align}\label{eqn2}
 \min_{\mathcal{A} \in \mathbb{A}}I(V;Y_\mathcal{A})-\max_{\mathcal{U} \in \mathbb{U}}I(V;Y_\mathcal{U}),
 \end{align}
$\delta_{\epsilon}^3(N)$ vanishes to zero in \eqref{nov11_14_1}, \eqref{c1eeee}, and the asymptotic public communication rate is as close as desired to
\begin{align}\label{eqn3}
\max_{\mathcal{A} \in \mathbb{A}}I(V;X|Y_\mathcal{A}) .
\end{align}
 By taking the auxiliary random variable  $V$ jointly Gaussian with $X$ in (\ref{eqn2}) and (\ref{eqn3}), we obtain~\eqref{eqf1} and \eqref{eqf2}, as shown in Appendix \ref{App_GA}. 
 
\begin{remark}
We observe that the size of the shares scales linearly with the secret size. First, note that the size of each share is the sum of the length of the public communication, i.e., $NR_p$ bits, and the length of $N$ quantized observations of a Gaussian random variable. Then, since we achieve the secret rate in \eqref{eqn2} by making the quantization parameter $l$ fixed when $N$ grows to infinity, we conclude that the size of the shares scales linearly with $N$, which is also the case for the length of the generated secret.
\end{remark} 
\section{Concluding Remarks}\label{sec:cm}
We studied {\color{black}information-theoretic} secret sharing from Gaussian correlated sources over a one-way rate-limited public channel and characterized its secret capacity, which provides a closed-form expression of the trade-off between public communication and the secret rate. By contrast with a traditional secret-sharing protocol, our setting does not require information-theoretically secure channels between the dealer and participants, and provides information-theoretic security during the distribution phase, where the dealer distributes shares of the secret to the participants. Moreover, we have shown that the size of the shares scales linearly with the size of the secret for any access structure. 
We also characterized the secret capacity for threshold access structures and showed that the secret capacity is, in general, not a monotone function of the threshold.

While explicit and low-complexity coding schemes have been proposed for {\color{black} information-theoretic secret sharing} that rely on discrete channel models \cite{chou2020,chou2018explicit} and discrete source models \cite{sultana2021}, developing low-complexity coding schemes that achieve the limits derived in this paper for Gaussian sources remains an open problem.
\appendices
\section{Derivation of \eqref{aug_291}, \eqref{aug_292}}\label{App_th1}
Let  $Z$ and $Z'$ be zero-mean jointly Gaussian and jointly non-singular random vectors with covariance matrices $\Sigma_Z$ and $\Sigma'_Z$, respectively. By \cite[Theorem 3.5.2]{gallager}, we have
\begin{align}
Z'= PZ+W, \label{dec7_1}
\end{align}
where $P \triangleq \Sigma_{Z'Z}\Sigma^{-1}_{Z}$ and $W$ is independent of $Z$ with covariance $\Sigma_{W}  \triangleq \Sigma_{Z'}-\Sigma_{Z'Z}\Sigma^{-1}_{Z}\Sigma^T_{Z'Z}$. 
Hence, by (\ref{dec7_1}), we have for any $\mathcal{S} \subseteq \mathcal{L}$
\begin{align}
Y_{\mathcal{S}}&=\Sigma_{Y_{\mathcal{S}}X}\sigma^{-2}_{X} X+W_{Y_\mathcal{S}}, \label{dec4_2}
\end{align}
where $\Sigma_{W_{Y_{\mathcal{S}}}}\triangleq\Sigma_{Y_\mathcal{S}}-\Sigma_{Y_{\mathcal{S}}X}\sigma^{-2}_X\Sigma^T_{Y_{\mathcal{S}}X}$. Then, we normalize~\eqref{dec4_2} as follows. By Cholesky decomposition, 
there exists an invertible matrix $B\in \mathbb{R}^{|\mathcal{S}|\times |\mathcal{S}|}$  such that $\Sigma_{W_{Y_{\mathcal{\mathcal{S}}}}}=BB^T$. Hence, 
\eqref{dec4_2} can be rewritten as 
\begin{align*}
Y'_{\mathcal{S}}=H_{\mathcal{S}}X+W'_{Y_{\mathcal{S}}},
\end{align*}
where $Y'_{\mathcal{S}} \triangleq B^{-1} Y_{\mathcal{S}}$, $H_{\mathcal{S}}=B^{-1}\Sigma_{Y_\mathcal{S}X}\sigma^{-2}_X$, and  $W'_{Y_{\mathcal{S}}}\sim\mathcal{N}(0,I_{|\mathcal{S}|})$.

\section{Proof of Theorem \ref{EX3}} \label{App_C}
To prove Theorem \ref{EX3}, we proceed as follows. For a threshold access structure $\mathbb{A}_t$, we first prove that there exist sets of authorized and unauthorized participants $\mathcal{A}_{t}^\star\in \argmin_{\mathcal{A}\in \mathbb{A}_{t}}H_\mathcal{A}^TH_\mathcal{A}$ and $\mathcal{U}_{t}^\star \in \argmax_{\mathcal{U}\in \mathbb{U}_{t}} H_\mathcal{U}^TH_\mathcal{U}$, respectively, such that  for any $t\in \llbracket 1 ,L-1 \rrbracket$,  $\mathcal{A}_{t}^\star\subset \mathcal{A}_{t+1}^\star$, $\mathcal{U}_{t}^\star\subset \mathcal{U}_{t+1}^\star$. 
Then, by Theorem \ref{theo1}, we remark that  $\mathcal{A}_{t}^\star$ and $\mathcal{U}_{t}^\star$ also correspond to the sets that appear in the expression of the secret capacity for the threshold access structure $\mathbb{A}_t$.  
Finally, using the monotonicity property (with respect to $t$) of the sets $(\mathcal{A}_{t}^\star)_{t \in \llbracket 1 , L \rrbracket}$ and $(\mathcal{U}_{t}^\star)_{t \in \llbracket 1 , L \rrbracket}$  and Theorem \ref{theo1}, we derive necessary and sufficient conditions to determine whether the secret capacity increases or decreases as the threshold $t$ increases.

We will need the following lemma.
\begin{lemma}\label{lem_ma_1}
Let $a, c \in \mathbb{R}_+$ and $R_p \in \mathbb{R}_+ $. The function $f_{a,c,R_p}$ is non-increasing
\begin{align*}
 f_{a,c, R_p}:   \mathbb{R}_+&\rightarrow \mathbb{R}\\
  y & \mapsto \frac{1}{2 }\log\frac{cy 2^{-2R_p}+c a (1-2^{-2R_p})+1}{cy+1}.
\end{align*}
\end{lemma}
\begin{proof}
The derivative of $f_{a,c, R_p}$ at $y\in \mathbb{R}_+$ is $f'_{a,c, R_p}=\frac{1}{2\ln 2}\frac{c(1+ca)(2^{-2R_p}-1)}{(cy+1)\left(cy 2^{-2R_p}+c a (1-2^{-2R_p})+1\right)}\leq 0$.
\end{proof}
Using Lemma \ref{lem_ma_1}, we obtain the following result:
 \begin{lemma}\label{prep2}
One can find sets $(\mathcal{A}_t^\star)_{t\in \llbracket 1 ,L \rrbracket}$ and $(\mathcal{U}_t^\star)_{t\in \llbracket 1 ,L \rrbracket}$ such that for any $t\in \llbracket 1 ,L-1 \rrbracket$, we have $\mathcal{A}_{t}^\star\subset \mathcal{A}_{t+1}^\star$, $\mathcal{U}_{t}^\star\subset \mathcal{U}_{t+1}^\star$, and for any $t\in \llbracket 1 ,L \rrbracket$,
\begin{align}
\{\mathcal{A}_{t}^\star,\mathcal{U}_{t}^\star\}\in & \argmin_{{{\mathcal{A} \in \mathbb{A}_{t}},{\mathcal{U}\in \mathbb{U}_{t}}}}\left[ f_{H_\mathcal{A}^T H_\mathcal{A},\sigma^2_X,R_p} (H^T_\mathcal{U} H_\mathcal{U})\right]^+, \label{eqargl}
\end{align}
where we have used the notation of Lemma \ref{lem_ma_1}.
\end{lemma}
\begin{proof}

For $t \in \llbracket 1 , L \rrbracket$, remark that 
\begin{align}
  \argmin_{{{\mathcal{A} \in \mathbb{A}_{t}},{\mathcal{U}\in \mathbb{U}_{t}}}}\left[ f_{H_\mathcal{A}^T H_\mathcal{A},\sigma^2_X,R_p} (H^T_\mathcal{U} H_\mathcal{U})\right]^+\nonumber\\=\left\{\argmin_{\mathcal{A} \in \mathbb{A}_{t}} H_\mathcal{A}^T H_\mathcal{A}, \argmax_{\mathcal{U} \in \mathbb{U}_{t}} H^T_\mathcal{U} H_\mathcal{U}\right\}, \label{eqargs}
\end{align}
  because $f_{H_\mathcal{A}^T H_\mathcal{A},\sigma^2_X,R_p} (H^T_\mathcal{U} H_\mathcal{U})$ is an increasing function of $H_\mathcal{A}^T H_\mathcal{A}$ and is a decreasing function of $H_\mathcal{U}^T H_\mathcal{U}$ by Lemma \ref{lem_ma_1}. Next, write the vector $H_{\mathcal{L}}$ as  $H_{\mathcal{L}}= [H_{\mathcal{L}}(1), H_{\mathcal{L}}(2), \dots , H_{\mathcal{L}}(L)]^T$. By relabelling the participants, if necessary, assume that $\vert H_{\mathcal{L}}(1)\vert \leq \vert H_{\mathcal{L}}(2)\vert \leq \dots \leq \vert H_{\mathcal{L}}(L)\vert$. For $t \in \llbracket 1 , L \rrbracket$, choose $\mathcal{A}_{t}^\star \triangleq \llbracket 1, t \rrbracket$ and $\mathcal{U}_{t}^\star \triangleq \llbracket L-t+2, L \rrbracket$. Clearly, for any $t\in \llbracket 1 ,L-1 \rrbracket$, we have $\mathcal{A}_{t}^\star\subset \mathcal{A}_{t+1}^\star$, $\mathcal{U}_{t}^\star\subset \mathcal{U}_{t+1}^\star$, and by \eqref{eqargs}, we have that \eqref{eqargl} holds for any $t\in \llbracket 1 ,L \rrbracket$.
\end{proof}

By Theorem \ref{theo1} and \eqref{eqargl}, we have 
\begin{align}
  C_s(\mathbb{A}_{1}, R_p)=\Bigg[\frac{1}{2}\log  \left( \sigma^2_X H_{\mathcal{A}_1^\star}^T H_{\mathcal{A}_1^\star}  (1-2^{-2R_p})+1 \right)\Bigg]^+, \label{eqtht1}
\end{align}
and for $t \in \llbracket 2, L \rrbracket$,  we have
\begin{align}
 & C_s(\mathbb{A}_{t}, R_p)=\nonumber\\&\Bigg[\frac{1}{2}\log\frac{\sigma^2_XH^T_{\mathcal{U}_{t}^\star} H_{\mathcal{U}_{t}^\star} 2^{-2R_p}+\sigma^2_X H_{\mathcal{A}_{t}^\star}^T H_{\mathcal{A}_{t}^\star}  (1-2^{-2R_p})+1}{\sigma^2_XH^T_{\mathcal{U}_{t}^\star} H_{\mathcal{U}_{t}^\star}+1}\Bigg]^+\!\!\!\!.  \label{eqtht}
\end{align}\\
Using \eqref{eqtht1} and \eqref{eqtht}, we easily obtain for any $t\in \llbracket 1, L \rrbracket$
\begin{align*}
&C_s(\mathbb{A}_{1},R_p)\geq C_s(\mathbb{A}_{t},R_p)\\
 &    \iff \sigma^2_X H_{\mathcal{A}_{1}^\star}^T H_{\mathcal{A}_{1}^\star} H^T_{\mathcal{U}_{t}^\star} H_{\mathcal{U}_{t}^\star}+H_{\mathcal{A}_{1}^\star}^T H_{\mathcal{A}_{1}^\star}+H^T_{\mathcal{U}_{t}^\star} H_{\mathcal{U}_{t}^\star}\nonumber\\&\hspace{5.2cm}-H_{\mathcal{A}_{t}^\star}^T H_{\mathcal{A}_{t}^\star}\geq 0.
 \end{align*} 
 From the proof of  Lemma \ref{prep2}, there exists $O\geq 0$ such that $O\leq H^T_{\mathcal{U}_{t}^\star} H_{\mathcal{U}_{t}^\star}$
and $ H_{\mathcal{A}_{1}^\star}^T H_{\mathcal{A}_{1}^\star}+O= H_{\mathcal{A}_{t}^\star}^T H_{\mathcal{A}_{t}^\star}$. Therefore,
$H_{\mathcal{A}_{1}^\star}^T H_{\mathcal{A}_{1}^\star}+H^T_{\mathcal{U}_{t}^\star} H_{\mathcal{U}_{t}^\star}\geq H_{\mathcal{A}_{t}^\star}^T H_{\mathcal{A}_{t}^\star}$, and $C_s(\mathbb{A}_{1},R_p)\geq C_s(\mathbb{A}_{t},R_p)$.

Next, we have for $i \in \llbracket 1 , L-t \rrbracket $,
 \begin{align*}
& C_s(\mathbb{A}_{t},R_p)\geq C_s(\mathbb{A}_{t+i},R_p) \nonumber \\ \nonumber
& \iff   \sigma^2_X H^T_{\mathcal{A}_{t}^\star}H_{\mathcal{A}_{t}^\star} H^T_{\mathcal{U}_{t+i}^\star}H_{\mathcal{U}_{t+i}^\star}+H^T_{\mathcal{A}_{t}^\star}H_{\mathcal{A}_{t}^\star}+H^T_{\mathcal{U}_{t+i}^\star}H_{\mathcal{U}_{t+i}^\star}\\
& \phantom{---}\geq \sigma^2_X H^T_{\mathcal{U}_{t}^\star}H_{\mathcal{U}_{t}^\star} H^T_{\mathcal{A}_{t+i}^\star}H_{\mathcal{A}_{t+i}^\star}\!+ H^T_{\mathcal{U}_{t}^\star}H_{\mathcal{U}_{t}^\star}\!+\! H^T_{\mathcal{A}_{t+i}^\star}H_{\mathcal{A}_{t+i}^\star}\\
& \iff \sigma^2_XH^T_{\mathcal{A}_{t}^\star}H_{\mathcal{A}_{t}^\star} (H^T_{\mathcal{U}_{t+i}^\star}H_{\mathcal{U}_{t+i}^\star} - H^T_{\mathcal{U}_{t}^\star}H_{\mathcal{U}_{t}^\star} +H^T_{\mathcal{U}_{t}^\star}H_{\mathcal{U}_{t}^\star} )\\&\hspace{5cm}+ H^T_{\mathcal{U}_{t+i}^\star}H_{\mathcal{U}_{t+i}^\star} - H^T_{\mathcal{U}_{t}^\star}H_{\mathcal{U}_{t}^\star} \\
& \phantom{---}\geq \sigma^2_X H^T_{\mathcal{U}_{t}^\star}H_{\mathcal{U}_{t}^\star} (H^T_{\mathcal{A}_{t+i}^\star}H_{\mathcal{A}_{t+i}^\star} - H^T_{\mathcal{A}_{t }^\star}H_{\mathcal{A}_{t }^\star} + H^T_{\mathcal{A}_{t }^\star}H_{\mathcal{A}_{t }^\star} )\\&\hspace{4.9cm}+ H^T_{\mathcal{A}_{t+i}^\star}H_{\mathcal{A}_{t+i}^\star} - H^T_{\mathcal{A}_{t }^\star}H_{\mathcal{A}_{t }^\star}\\
& \iff (1+\sigma^2_XH^T_{\mathcal{A}_{t}^\star}H_{\mathcal{A}_{t}^\star}) (H^T_{\mathcal{U}_{t+i}^\star}H_{\mathcal{U}_{t+i}^\star} - H^T_{\mathcal{U}_{t}^\star}H_{\mathcal{U}_{t}^\star} ) \\
& \phantom{---}\geq (1+\sigma^2_X H^T_{\mathcal{U}_{t}^\star}H_{\mathcal{U}_{t}^\star}) (H^T_{\mathcal{A}_{t+i}^\star}H_{\mathcal{A}_{t+i}^\star} - H^T_{\mathcal{A}_{t }^\star}H_{\mathcal{A}_{t }^\star} ),
 \end{align*}
 where the first equivalence is obtained using \eqref{eqtht}. 
Note that, by Lemma \ref{prep2}, one can choose $\mathcal{A}_{t}^\star\subset\mathcal{A}_{t+i}^\star$ and  $\mathcal{U}_{t}^\star\subset\mathcal{U}_{t+i}^\star$, hence, $H^T_{\mathcal{A}_{t+i}^\star}H_{\mathcal{A}_{t+i}^\star} - H^T_{\mathcal{A}_{t }^\star}H_{\mathcal{A}_{t }^\star} \geq 0$ and $H^T_{\mathcal{U}_{t+i}^\star}H_{\mathcal{U}_{t+i}^\star} - H^T_{\mathcal{U}_{t}^\star}H_{\mathcal{U}_{t}^\star}  \geq~0$.

\section{Proof of \eqref{eqerror}} \label{Appeq10}
The probability of error averaged over $C_n$, i.e., $\mathbb{E}_{C_n}\left[\mathbb{P}[V^n\neq \widehat{V}^n_\mathcal{A}]\right]$ for any $\mathcal{A} \in \mathbb{A}$ can be upper bounded via the union bound by the four following terms: 
\begin{enumerate}
    \item The probability that  $(x^n, y_\mathcal{A}^n) \notin{\mathcal{T}_{\epsilon_1}^n(XY_\mathcal{A})}$, which is upper bounded by $2 \vert \mathcal{X}\vert\vert\mathcal{Y}_\mathcal{A}\vert\exp(-n\epsilon_1^2\mu_{XY_\mathcal{A}})$ \cite[Page~272 Equation (1.12)]{r4}.
    \item The probability that the encoder cannot find $(\omega,\nu)$ such that $(x^n,v^n(\omega,\nu))\in \mathcal{T}_{\epsilon}^n(XV)$, given that $(x^n, y_\mathcal{A}^n) \in{\mathcal{T}_{\epsilon_1}^n(XY_\mathcal{A})}$, which is upper bounded by 
    \begin{align*}
&\mathbb{E}_{C_n} \Bigg[\sum_{x^n,{y}_\mathcal{A}^n}p_{X^nY_\mathcal{A}^n}(x^n,y_\mathcal{A}^n){\mathbbm{1}}\{\forall(\omega,\nu), (v^n(\omega,\nu),x^n)\\&\hspace{1.7cm}\notin {\mathcal{T}_{\epsilon}^n(VX)}\text{ and}~(x^n,y_\mathcal{A}^n) \in {\mathcal{T}_{\epsilon_1}^n(XY_\mathcal{A})}\} \Bigg] \nonumber \\
&= \sum_{(x^n,{y}_\mathcal{A}^n) \in T^n_{\epsilon_1}(XY_\mathcal{A})}\!p_{X^nY_\mathcal{A}^n}(x^n,y_\mathcal{A}^n)\mathbb{P}[\forall (\omega,\nu),\\&\hspace{3.9cm}(V^n(\omega,\nu),x^n)\notin {\mathcal{T}_{\epsilon}^n(VX)}] \nonumber \\
&= \sum_{(x^n,{y}_\mathcal{A}^n) \in T^n_{\epsilon_1}(XY_\mathcal{A})}p_{X^nY_\mathcal{A}^n}(x^n,y_\mathcal{A}^n)(1\\&\hspace{1.5cm}-\mathbb{P}[(V^n(\omega,\nu),x^n)\in {\mathcal{T}_{\epsilon}^n(VX)}])^{2^{n(R_v+R'_v)}} \nonumber \\
& \stackrel{(a)}\leq \sum_{(x^n,{y}_\mathcal{A}^n) \in T^n_{\epsilon_1}(XY_\mathcal{A})}p_{X^nY_\mathcal{A}^n}(x^n,y_\mathcal{A}^n)\exp(\!-{2^{n(R_v+R'_v)}}\\&\hspace{3cm}\times \mathbb{P}[(V^n(\omega,\nu),x^n)\in {\mathcal{T}_{\epsilon}^n(VX)}])\nonumber \\
& \stackrel{(b)} \leq \sum_{(x^n,{y}_\mathcal{A}^n) \in T^n_{\epsilon_1}(XY_\mathcal{A})} \!\! p_{X^nY_\mathcal{A}^n}(x^n,y_\mathcal{A}^n)\exp\Big(-{2^{n(R_v+R'_v)}} \\&\hspace{2.2cm}\times\left(1-\delta_{\epsilon_1,\epsilon}^{(2)}(n)\right)2^{-n(I(V;X)+2\epsilon H(V))}\Big)\nonumber \\
&\leq \exp\Big(-\left(1-\delta_{\epsilon_1,\epsilon}^{(2)}(n)\right)2^{\epsilon nH(V)}\Big),\nonumber 
\end{align*}
where $(a)$ holds because for any $x\geq 0$ and any $p \in [0,1]$, $(1-p)^x \leq e^{-px}$, and in $(b)$ we have defined $\delta^{(2)}_{\epsilon_1,\epsilon}(n)\triangleq 2\vert\mathcal{V}\vert \vert\mathcal{X}\vert\exp\left(-n\frac{(\epsilon-\epsilon_1)^2}{1+\epsilon_1}\mu_{VX}\right)$.
\item The probability that the decoder finds $\tilde{\nu}_{\mathcal{A} } \neq \nu$ such that $(y_{\mathcal{A} }^n,v^n(\omega,\tilde{\nu}_{\mathcal{A} }))\in \mathcal{T}_{\epsilon}^n(Y_{\mathcal{A} }V)$, given that $(x^n, y_\mathcal{A}^n) \in{\mathcal{T}_{\epsilon_1}^n(XY_\mathcal{A})}$ and the encoder found $(\omega,\nu)$ such that $(x^n,v^n(\omega,\nu))\in \mathcal{T}_{\epsilon}^n(XV)$, which is upper bounded by 
\begin{align*}
&\sum_{\omega,\nu}p(\omega,\nu)\sum_{\nu'_\mathcal{A}\neq \nu}\mathbb{E}_{C_n} \sum_{(x^n,{y}_\mathcal{A}^n) \in T^n_{\epsilon_1}(XY_\mathcal{A})}p_{X^nY_\mathcal{A}^n}(x^n,y_\mathcal{A}^n)\\&\hspace{2.7cm}\times{\mathbbm{1}}\{y_\mathcal{A}^n,(v^n(\omega,\nu'_\mathcal{A}))\in {\mathcal{T}_{\epsilon}^n(Y_\mathcal{A}V)}\}  \nonumber \\
&=\sum_{\omega,\nu}p(\omega,\nu)\sum_{\nu'_\mathcal{A}\neq \nu}\sum_{(x^n,{y}_\mathcal{A}^n) \in T^n_{\epsilon_1}(XY_\mathcal{A})}p_{X^nY_\mathcal{A}^n}(x^n,y_\mathcal{A}^n)\\&\hspace{2.7cm}\times\mathbb{P}[(y_\mathcal{A}^n,(V^n(\omega,\nu'_\mathcal{A}))\in {\mathcal{T}_{\epsilon}^n(Y_\mathcal{A}V)}] \nonumber \\
&\leq\sum_{\omega,\nu}p(\omega,\nu)\sum_{\nu'_\mathcal{A}\neq \nu}\sum_{(x^n,{y}_\mathcal{A}^n) \in T^n_{\epsilon_1}(XY_\mathcal{A})}p_{X^nY_\mathcal{A}^n}(x^n,y_\mathcal{A}^n)\\&  \phantom{----------l--}\times2^{-n(I(V;Y_\mathcal{A})-2\epsilon H(V))} \nonumber  \displaybreak[0]\\
&\leq2^{n(R_v'-I(V;Y_\mathcal{A})+2\epsilon H(V))}\nonumber \displaybreak[0]\\
&\leq 2^{-n\epsilon H(V)}.\nonumber 
\end{align*}
\item The probability that the decoder cannot find $\tilde{\nu}_{\mathcal{A} } $ such that $(y_{\mathcal{A} }^n,v^n(\omega,\tilde{\nu}_{\mathcal{A} }))\in \mathcal{T}_{\epsilon}^n(Y_{\mathcal{A} }V)$, given that $(x^n, y_\mathcal{A}^n) \in{\mathcal{T}_{\epsilon_1}^n(XY_\mathcal{A})}$ and the encoder found $(\omega,\nu)$ such that $(x^n,v^n(\omega,\nu))\in \mathcal{T}_{\epsilon}^n(XV)$, which is upper bounded with Markov lemma \cite[Page 319  Equation (5.1)]{r4} by $2 \vert \mathcal{V}\vert\vert \mathcal{X}\vert\vert\mathcal{Y}_\mathcal{A}\vert\exp\left(-n\frac{(\epsilon-\epsilon_1)^2}{1+\epsilon_1}\mu_{VXY_\mathcal{A}}\right)$ .
\end{enumerate}

Hence, for any $\mathcal{A} \in \mathbb{A}$,  we have $\mathbb{E}_{C_n}[\mathbb{P}[V^n\neq \widehat{V}^n_\mathcal{A}]]\leq \delta(n,\epsilon,\mathcal{A})$. Next, we have
\begin{align}
\mathbb{E}_{C_n}\left[\max_{\mathcal{A} \in \mathbb{A}}\mathbb{P}[\widehat{V}_\mathcal{A}^n\neq V^n]\right]&\leq \mathbb{E}_{C_n} \left[  \sum_{\mathcal{A}\in \mathbb{A}} \mathbb{P}[\widehat{V}_\mathcal{A}^n \neq V^n] \right] \nonumber\\
&= \sum_{\mathcal{A} \in \mathbb{A}} \mathbb{E}_{C_n} \left[ \mathbb{P}[\widehat{V}_\mathcal{A}^n \neq V^n] \right] \nonumber\\
&\leq \sum_{\mathcal{A} \in \mathbb{A}} \delta(n,\epsilon,\mathcal{A})\nonumber\\
&\leq \vert\mathbb{A} \vert\max_{\mathcal{A} \in \mathbb{A}} \delta(n,\epsilon,\mathcal{A}).\nonumber
\end{align}
By Markov's inequality, we conclude that there exists a codebook such that $\max_{\mathcal{A} \in \mathbb{A}}\mathbb{P}[\widehat{V}_\mathcal{A}^n \neq V^n] \leq \vert\mathbb{A} \vert\max_{\mathcal{A} \in \mathbb{A}}\delta(n,\epsilon,\mathcal{A})$. 

\section{Proof of \eqref{rl6}} \label{App16}
For any $\mathcal{U}\in \mathbb{U}$, we have
\begin{align}
H(V^n\vert Y_\mathcal{U}^n)
& \stackrel{(a)}{\geq} I(X^n;V^n\vert Y_\mathcal{U}^n)\nonumber\\
&=H(X^n\vert Y_\mathcal{U}^n)-H(X^n\vert V^n Y_\mathcal{U}^n) \nonumber\\
& \stackrel{(b)}{=} nH(X\vert Y_\mathcal{U})-H(X^n\vert V^n Y_\mathcal{U}^n),
\label{rl3}
\end{align}
where $(a)$ holds by definition of mutual information, and $(b)$ holds because the $X_i$'s and $(Y_\mathcal{U})_i$'s are independently and identically distributed.
We now lower bound the term $-H(X^n\vert V^n Y_\mathcal{U}^n)$. Define for any $\mathcal{U} \in \mathbb{U}$,
\begin{align*}
\Gamma_{\mathcal{U}} & \triangleq 
\mathbbm{1}\{(X^n, V^n, Y_\mathcal{U}^n) \in \mathcal{T}_{2\epsilon}^n(XV Y_\mathcal{U})\} ,  \\
\Delta_{\mathcal{U}} &\triangleq
\mathbbm{1}\{ (X^n, V^n) \in \mathcal{T}_{\epsilon}^n(XV )\},
\end{align*}
so that,
\begin{align}
&H(X^n\vert V^n Y_\mathcal{U}^n) \nonumber \\
&\leq H(X^n\Gamma_{\mathcal{U}} \Delta_{\mathcal{U}}\vert V^n Y_\mathcal{U}^n) \nonumber\\
&= H(\Gamma_{\mathcal{U}} \Delta_{\mathcal{U}}\vert V^n Y_\mathcal{U}^n)+H(X^n\vert V^n Y_\mathcal{U}^n\Gamma_{\mathcal{U}} \Delta_{\mathcal{U}})\nonumber \\
& \stackrel{(a)}\leq 2+\smash{\sum_{\delta_\mathcal{U},\gamma_\mathcal{U} \in \{0,1\}}}\mathbb{P}(\Gamma_{\mathcal{U}}=\gamma_{\mathcal{U}}\vert\Delta_{\mathcal{U}}=\delta_{\mathcal{U}})\mathbb{P}(\Delta_{\mathcal{U}}=\delta_{\mathcal{U}})\nonumber\\&\hspace{3.1cm}\times H(X^n\vert V^n Y_\mathcal{U}^n,\Gamma_{\mathcal{U}}=\gamma_{\mathcal{U}}, \Delta_{\mathcal{U}}=\delta_{\mathcal{U}}) \nonumber\\
& \stackrel{(b)}\leq 2+H(X^n\vert V^n Y_\mathcal{U}^n,\Gamma_{\mathcal{U}}=1, \Delta_{\mathcal{U}}=1)\nonumber\\&\hspace{3.8cm}+(2\delta_{\epsilon}(n)+\delta^2_{\epsilon}(n, \mathcal{U}))\log\vert\mathcal{X}\vert^n \nonumber \\ \nonumber
&=\smash{\sum_{y_\mathcal{U}^n,v^n}}p(y_\mathcal{U}^n,v^n\vert {1,1})\\
&\hspace{1.4cm} \times H(X^n\vert  Y_\mathcal{U}^n=y_\mathcal{U}^n,V^n=v^n,\Gamma_{\mathcal{U}}=1\nonumber,\Delta_{\mathcal{U}}=1) \nonumber \\
& \hspace{3cm}+ 2 +(2\delta_{\epsilon}(n)+\delta^2_{\epsilon}(n, \mathcal{U}))\log\vert\mathcal{X}\vert^n\nonumber\\
& \stackrel{(c)} \leq \smash{\sum_{y_\mathcal{U}^n,v^n}}p(y_\mathcal{U}^n,v^n\vert {1,1})\log\lvert T^n_{2\epsilon}(X\vert y_\mathcal{U}^n,v^n)\rvert \nonumber\\&\hspace{3cm}+ 2 +(2\delta_{\epsilon}(n)+\delta^2_{\epsilon}(n, \mathcal{U}))\log\vert\mathcal{X}\vert^n \nonumber \displaybreak[0]\\
&\leq \smash{\sum_{y_\mathcal{U}^n,v^n}}p(y_\mathcal{U}^n,v^n\vert {1,1})nH(X\vert Y_\mathcal{U}V)(1+2\epsilon)  \nonumber\\&\hspace{3cm}+ 2+(2\delta_{\epsilon}(n) +\delta^2_{\epsilon}(n, \mathcal{U}))\log\vert\mathcal{X}\vert^n \nonumber \displaybreak[0]\\
&\leq nH(X\vert Y_\mathcal{U}V)(1+2\epsilon) \!+ 2 +(2\delta_{\epsilon}(n)\!+\delta^2_{\epsilon}(n, \mathcal{U}))\log\vert\mathcal{X}\vert^n.
\label{rl5}
\end{align}
where $(a)$ holds because $(\Gamma_{\mathcal{U}}, \Delta_{\mathcal{U}})$ is defined over an alphabet of cardinality equal to four so that $H(\Gamma_{\mathcal{U}} \Delta_{\mathcal{U}}\vert V^n Y_\mathcal{U}^n) \leq \log 4 =2$, $(b)$ holds because $\mathbb{P}[\Delta_\mathcal{U}=0] \leq \delta_{\epsilon}(n)\triangleq {\color{black}2 \vert \mathcal{X}\vert\vert\mathcal{V}\vert e^{-n\epsilon^2\mu_{XV}}}$ and $\mathbb{P}[\Gamma_\mathcal{U}=0\vert \Delta_\mathcal{U}=1] \leq \delta^2_{\epsilon}(n, \mathcal{U})\triangleq {\color{black}2\vert \mathcal{V}\vert\vert\mathcal{X}\vert\vert\mathcal{Y}_\mathcal{U}\vert e^{-\epsilon^2n\mu_{VXY_\mathcal{U}}/8}}$ 
by Markov Lemma \cite[Page~319  Equation (5.1)]{r4}, and $(c)$ holds because $H(X) \leq \log  |\mathcal{X}|$ for any discrete random variable $X$ defined over $ |\mathcal{X}|$. 
Combining~(\ref{rl3}) and (\ref{rl5}), we obtain \eqref{rl6}.
\section{Proof of \eqref{eqf1} and \eqref{eqf2}} \label{App_GA}
We rewrite (\ref{eqn2}) and (\ref{eqn3}) as
 \begin{align}\label{sep5_3}
 R_p&=\max_{\mathcal{A} \in \mathbb{A}}\left(h(X)-h(X\vert V)-h(Y_\mathcal{A})+h(Y_\mathcal{A}\vert V)\right),\\
 R_s&=\min_{\mathcal{A} \in \mathbb{A}} \min_{\mathcal{U} \in \mathbb{U}} \left(h(Y_\mathcal{A})-h(Y_\mathcal{A}\vert V)- h(Y_\mathcal{U})+h(Y_\mathcal{U}\vert V)\right).\label{sep5_4}
 \end{align}
Let $K_{XV}\triangleq\left[\begin{matrix}
\sigma^2_{X}&\sigma_{XV}\\
\sigma_{VX}&{\sigma^2_V}
\end{matrix}\right]$ be the covariance matrix of $(X,V)$. We have
\begin{align}
h(X\vert V)&=h(X,V)-h(V)\nonumber\\
&=\frac{1}{2}\log(2\pi e)^2\det( K_{XV})-\frac{1}{2}\log2\pi e{\sigma^2_V}\nonumber\\
& =\frac{1}{2}\log2\pi e ( \sigma^2_{X}-\sigma_{XV}\sigma^{-2}_V\sigma_{XV})\nonumber\\
& =\frac{1}{2}\log2\pi e \sigma^2_{X \vert V},\label{aug1}
\end{align}
where the last equality holds by \cite[Proposition 3.13]{eaton}. 
Next,  for any $\mathcal{A} \in \mathbb{A}$, let $K_{Y_\mathcal{A}V}\triangleq\left[\begin{matrix}
\Sigma_{Y_\mathcal{A}} &\Sigma_{Y_\mathcal{A}V}\\
\Sigma_{Y_\mathcal{A}V}^T&{\sigma^2_V}
\end{matrix}\right]$ be the covariance matrix of $(Y_\mathcal{A},V)$. We have  
\begin{align}
h(Y_\mathcal{A}\vert V)
&=\frac{1}{2}\log(2\pi e)^{\vert \mathcal{A}\vert}\frac{\det( K_{Y_\mathcal{A}V})}{{\sigma^2_V}}\nonumber\\
&\stackrel{(a)}{=}\frac{1}{2}\log(2\pi e)^{\vert \mathcal{A}\vert}\frac{{\sigma^2_V}\det( \Sigma_{Y_\mathcal{A}}-\Sigma_{Y_\mathcal{A}V}\sigma^{-2}_V\Sigma_{Y_\mathcal{A}V}^{T})}{{\sigma^2_V}}\nonumber\\
&\stackrel{(b)}{=}\frac{1}{2}\log(2\pi e)^{\vert \mathcal{A}\vert}\det( \Sigma_{Y_\mathcal{A}\vert V}),\nonumber\\
& \stackrel{(c)}{=}\frac{1}{2}\log(2\pi e)^{\vert \mathcal{A}\vert}\det( H_\mathcal{A}\sigma^2_{X \vert V}H_\mathcal{A}^T+I), \label{aug29_2} 
\end{align}
where $(a)$ holds by the formula for the determinant of a block matrix, $(b)$ holds by \cite[Proposition 3.13]{eaton}, $(c)$ holds by (\ref{aug_291}) and the definition of the conditional variance  $\Sigma_{Y_\mathcal{A} \vert V}\triangleq\mathbb{E}\left[\left(Y_\mathcal{A}-\mathbb{E}[Y_\mathcal{A}\vert V]\right)\left(Y_\mathcal{A}-\mathbb{E}[Y_\mathcal{A} \vert V]\right)^T\vert V\right]=H_\mathcal{A} \mathbb{E}\left[\left(X\!\!-\mathbb{E}[X\vert V]\right)\left(X\!\!-\mathbb{E}[X \vert V]\right)^T\vert V\right]\!H_\mathcal{A}^T +\mathbb{E}\left[W_{Y_\mathcal{A}} W^T_{Y_\mathcal{A}}\right]$ and $W_{Y_\mathcal{A}}$ is a Gaussian noise vector with  with identity covariance matrix. 
Similarly, for any $\mathcal{U} \in \mathbb{U}$, we have
\begin{align}
h(Y_\mathcal{U}\vert V)
& = \frac{1}{2}\log(2\pi e)^{\vert \mathcal{U}\vert}\det( H_\mathcal{U}\sigma^2_{X \vert V}H_\mathcal{U}^T+I). \label{sep9_3} 
\end{align}

Thus, from \eqref{sep5_3}, \eqref{sep5_4}, (\ref{aug1}), (\ref{aug29_2}), and (\ref{sep9_3}), we have
\begin{align}
&R_p=\max_{\mathcal{A} \in \mathbb{A}}\left(\frac{1}{2}\log{ \frac{\sigma^2_{X}}{ \sigma^2_{X \vert V}} }-\frac{1}{2}\log \frac{\det( H_\mathcal{A}\sigma^2_{X}H_\mathcal{A}^T+I)}{\det( H_\mathcal{A} \sigma^2_{X \vert V}H_\mathcal{A}^T+I)}\right), \label{aug29_3b} \\
&R_s =\min_{\mathcal{A} \in \mathbb{A}}\min_{\mathcal{U} \in \mathbb{U}}\Bigg(\frac{1}{2}\log \frac{\det( H_\mathcal{A}\sigma^2_{X}H^T_\mathcal{A}+I)}{\det( H_\mathcal{A}\sigma^2_{X \vert V}H^T_\mathcal{A}+I)}\nonumber\\&\hspace{2.9cm}-\frac{1}{2}\log \frac{\det( H_\mathcal{U}\sigma^2_{X}H^T_\mathcal{U}+I)}{\det( H_\mathcal{U}\sigma^2_{X \vert V}H^T_\mathcal{U}+I)}\Bigg).
\label{aug29_3}
\end{align}
Then, by Lemma \ref{Laug27_1} and the definition of $O_\mathcal{A}$, $\mathcal{A} \in \mathbb{A}$ and $O_\mathcal{U}$, $\mathcal{U} \in \mathbb{U}$, we can rewrite (\ref{aug29_3b}) and~\eqref{aug29_3} as 
\begin{align}
R_p&=\max_{\mathcal{A} \in \mathbb{A}}\left(\frac{1}{2}\log\frac{\sigma^2_{X}}{\sigma^2_{X \vert V}}-\frac{1}{2}\log\frac{\sigma^2_{X}O_\mathcal{A}+1}{\sigma^2_{X \vert V}O_\mathcal{A}+1}\right), \label{eqRp} \\
R_s &=  \min_{\mathcal{A} \in \mathbb{A}}\min_{\mathcal{U} \in \mathbb{U}}\left(\frac{1}{2}\log\frac{ \sigma^2_{X} O_\mathcal{A}+1}{\sigma^2_{X \vert V}O_\mathcal{A}+1}-\frac{1}{2}\log\frac{ \sigma^2_{X} O_\mathcal{U}+1}{\sigma^2_{X \vert V}O_\mathcal{U}+1}\right) . \label{eqRs}
\end{align}
Finally, by Lemma \ref{Ldec6_1}, \eqref{eqRp} and~\eqref{eqRs} become \eqref{eqf1} and \eqref{eqf2}.

\bibliographystyle{IEEEtran}
\bibliography{reff}

\end{document}